\documentclass[alpha-refs]{wiley-article}

\usepackage{color}
\usepackage{ulem}

\usepackage{amssymb}
\usepackage{amsmath}
\usepackage{graphicx}
\usepackage{xcolor}
\usepackage{float}
\usepackage{caption}
\usepackage{subcaption}
\PassOptionsToPackage{hyphens}{url}\usepackage[hidelinks]{hyperref}

\DeclareMathOperator*{\argmin}{arg\,min}

\title{Testing for no effect in regression problems: a permutation approach}

\author[1]{Micha{\l} G. Ciszewski}
\author[1]{Jakob S\"{o}hl}
\author[2]{Ton Leenen}
\author[3]{Bart van Trigt}
\author[1]{Geurt Jongbloed}

\affil[1]{Applied Mathematics, Delft University of Technology, Mekelweg 4, 2628 CD Delft, Netherlands}
\affil[2]{Faculty of Behavioural and Movement Sciences, VU University Amsterdam, Van Der Boechorststraat 7, 1081 BT Amsterdam, Netherlands}
\affil[3]{Biomechanical Engineering, Delft University of Technology, Mekelweg 2, 2628 CD Delft, Netherlands}

\corraddress{Micha{\l} G. Ciszewski, Applied Mathematics, Delft University of Technology, Mekelweg 4, 2628 CD Delft, Netherlands}
\corremail{M.G.Ciszewski@tudelft.nl}

\runningauthor{Ciszewski et al.}

\fundinginfo{\begin{minipage}[t]{.3\columnwidth}
		This work is part of the research programme CAS with project number P16-28 project 2, which is (partly) financed by the Dutch Research Council (NWO).
	\end{minipage}
	\hspace{0.02\linewidth}
	\begin{minipage}[t]{.6\linewidth}
		\centering\raisebox{\dimexpr \topskip-\height}{%
			\includegraphics[width=0.9\linewidth]{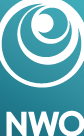}}
\end{minipage}}

\begin{document}

\maketitle

\begin{abstract}
Often the question arises whether $Y$ can be predicted based on $X$ using a certain model. 
Especially for highly flexible models such as neural networks one may ask whether a seemingly good prediction is actually better than fitting pure noise or whether it has to be attributed to the flexibility of the model. 
This paper proposes a rigorous permutation test to assess whether the prediction is better than the prediction of pure noise.
The test avoids any sample splitting and is based instead on generating new pairings of $(X_i,Y_j)$.
It introduces a new formulation of the null hypothesis and rigorous justification for the test, which distinguishes it from previous literature.
The theoretical findings are applied both to simulated data and to sensor data of tennis serves in an experimental context.
The simulation study underscores how the available information affects the test. It shows that the less informative the predictors, the lower the probability of rejecting the null hypothesis of fitting pure noise and emphasizes that detecting weaker dependence between variables requires a sufficient sample size.

\end{abstract}

\section{Introduction}

With the ubiquity of data often the question whether a response $Y$ can be predicted based on predictors $X$ arises.
The rise of highly capable machine learning and deep learning techniques increases the abilities to fit any kind of data. 
However, the abilities to fit pure noise are increasing as well. We propose a method to test whether a model is only fitting noise. 
It extends testing for no effect from linear to nonlinear models. 
No sample splitting is performed so the power of the test can rely on the size of the whole sample. 
No nested sequence of models is needed, in fact, no alternative models are needed at all.

Our method is based on recombining the pairings between predictors and responses through permutations.
In this way artificial reference datasets are created and the performance of the model on the original data can then be assessed by comparing it to the performances on the artificial reference datasets.
The purpose of our test is to ascertain whether the model is capable of fitting the data more effectively than mere random noise.
Our method is not restricted to linear models since it is not a test for specific parameters in the model. 
Rather it tests for the ability of a model to predict the responses.

The main contribution of this paper is a rigorous formulation of a permutation test for dependence between model predictions and responses.
The test uses $R^2$ as test statistic but can be performed with any measure of goodness of fit in regression analysis.
Because of its interpretability, $R^2$ is our test statistic of choice, but this can be adapted if necessary.
The method generates new pairings of $(X_i,Y_j)$ conditional on the $X_i$ for $i=1,...,n$ and $Y_j$ for $j=1,...,n$.
This paper introduces a new formulation of the null hypothesis and provides a rigorous justification for a permutation test that has been described in various forms in the literature,  for instance in the two-sample problem (\cite{good02, commenges03, huang06, hutson12}), the stochastic dominance problem (\cite{arboretti08, arboretti09a}) or the subgroup discovery problem (\cite{duivesteijn11}).
The main use case for this method is in the initial stages of the data analysis  to test whether a given model does only fit noise or is able to capture some essential structure in the data.

The outline of the paper is as follows.
Section \ref{sec:problem} formulates the problem and introduces necessary notation.
Section \ref{sec:permtest} contains the formal formulation of the null hypothesis, theoretical considerations and the succinct description of the method.
Section \ref{sec:application} contains the application of the permutation test to sensor data of tennis serves in order to demonstrate the method in practice and showcase its power in a real-life scenario.
Section \ref{sec:simstudy} presents a simulation study, where the permutation test is demonstrated in various scenarios for predictors and responses.

\section{Methodology}

\subsection{Problem description}\label{sec:problem}

Consider a regression setting. Given an observed pair $(X,Y)$, where $X$ is a random vector and $Y$ is a real random variable.
$Y$ is modelled as:
\begin{equation*}
	Y=f(X)+\epsilon,
\end{equation*}
where $\epsilon$ is a centered random variable independent of $X$ and $(f(X))_{f\in\mathcal{F}}$ for some class of functions $\mathcal{F}$.
An example of $\mathcal{F}$ could be a set of all linear functions corresponding to a linear regression model with fixed number of variables or a set of functions that can be described by a neural network.
Nonparametric classes of functions can also be considered, for instance a set of log-concave functions.
For the remainder of the paper, we will focus on $R^2$ as goodness-of-fit measure.

Since the actual relationship between $X$ and $Y$ is not known in practice, a chosen class of functions $\mathcal{F}$ through which that relationship is described does not need to be appropriate.
The class of functions $\mathcal{F}$ is misspecified if it does not contain the true $f$, while if it contains too many functions, the model might be overfitting by memorizing the noise $\epsilon$.
In a real world scenario, we are often facing datasets that feature high-dimensional, time-dependent or functional variables.
The question whether there is a relationship between $X$ and $Y$ and which model to choose for describing it, is crucial.
In this paper, we focus on the following aspect:
\begin{itemize}
	\item can a given class of functions $\mathcal{F}$ distinguish $Y$ from pure noise?
\end{itemize}

Consider this simple example.
Let $X_1, X_2$ be independent standard normal variables and $Y=X_1^2+X_2^2+\epsilon$, where $\epsilon\sim\mathcal{N}(0,0.01)$.
Consider a multi-layered neural net as a model of choice to predict $Y$ using $X_1$ and $X_2$.
For small sample sizes shuffling the vector of responses and applying our prediction model to this shuffled dataset can yield values of $R^2$ higher than values of $R^2$ calculated for the prediction model applied to the original dataset.
Ten random samples of size 10 were drawn.
Five yielded higher values of $R^2$ for at least one shuffled dataset than for the original pairing (we considered 200 shuffles of the sample).
In applied settings, where the sample size is fixed and difficult to increase, this presents an inherent issue.
Sample size has an immediate influence on the credibility of the model and needs to be taken into account.

Related problems have been addressed before in the literature in different settings and with a variety of solutions.
In this paper we focus on the permutation test.

\subsection{Permutation approach to testing for no effect}\label{sec:permtest}

The main appeal of permutation tests stems from the fact that they do not require any distributional assumptions on the population.
The lack of assumptions is increasingly more interesting to researchers as deep learning methods become more popular since they likewise do not rely on distributional assumptions.
Permutation tests are completely data-driven as pointed out by \cite{berry02}.
This can be very appealing as the data is the main factor in shaping the distribution of the test statistic, i.e. the test statistic can be chosen to be more easily interpretable without focusing on its distribution.

Our work differs from previous works mostly in the formulation of the null hypothesis.
In the literature different null hypotheses exist.
Some involve the concept of exchangeability, e.g. \cite{romano90, schmoyer94, good02, commenges03, huang06, hutson12}.
Some involve equality of means, e.g. \cite{zhang09} and some involve zeroing of the coefficients, e.g. \cite{cardot04}.
In contrast to this, we focus on the concept of independence, which is not widely used for permutation tests.
Permutation tests of independence have existed before, e.g. see \cite{bell67}.
However, we do not test independence of two random variables $X$ and $Y$, but rather we state the null hypothesis in terms of the model and whether it is able to capture the dependence.

The choice of the null hypothesis can also be directly connected to the model considered in the problem.
For instance, it is natural to use zeroing of the functional coefficient as the null hypothesis when considering a functional linear regression model e.g. see \cite{cardot04}.
We do not restrict ourselves to any particular model in our work, we only consider the model as given and the null hypothesis is not specifically tailored to the model.

Our goal is to investigate whether a class of functions $\mathcal{F}$ can capture any dependence between $X$ and $Y$.
We consider a test with null hypothesis stated as follows:
\begin{equation}\label{null}
	H_0: Y \textrm{ is independent of } (f(X))_{f\in\mathcal{F}}.
\end{equation}
$H_0$ represents the problem as described in section \ref{sec:problem}.
If it were true, then our class of functions $\mathcal{F}$ will not be able to capture the relation between $X$ and $Y$ in a meaningful manner.
Considering a dataset with permuted responses will be no different to class of functions $\mathcal{F}$ under $H_0$.
If $H_0$ is false, then the class of functions $\mathcal{F}$ will be able to capture some aspects of the relation between $X$ and $Y$, although it does not guarantee that the model is suitable and readily applicable.

The null hypothesis $H_0$ as stated in \eqref{null} does not guide the choice of the test statistic.
In order to choose a suitable test statistic, further understanding of $H_0$ is needed.
\begin{proposition}\label{perminvariance}
	Let $(X_1,Y_1),...,(X_n,Y_n)$ be an i.i.d. sample of $(X,Y)$.
	If $Y$ is independent of $(f(X))_{f\in\mathcal{F}}$, then for all $i=1,\dots,n$ the conditional distribution of
	\begin{equation}\label{eq:jointconddist}
		\big((f(X_i),Y_{\tau(i)})\big)_{f\in \mathcal{F}}
	\end{equation}
	given the empirical measure $P_n^X$ of $X_1,...,X_n$ and the empirical measure $P_n^Y$ of $Y_1,...,Y_n$ is the same for all permutations $\tau$ of set $\{1,...,n\}$.
\end{proposition}
\begin{proof}
	Let $i\in\{1,\dots,n\}$ be fixed.
	For a given finite collection of functions $f_1,f_2,...,f_m\in\mathcal{F}$ and a permutation $\tau$, the conditional joint distribution of $(f_1(X_i),Y_{\tau(i)}),...,(f_m(X_i),Y_{\tau(i)})$ given $P_n^X$  and $P_n^Y$ is the same as the joint distribution of
	\begin{equation}\label{eq:distspecial}
		(f_1(X_{i}),Y_{i}),..., (f_m(X_{i}),Y_{i}),
	\end{equation}
	thanks to the assumption of independence of $(f(X))_{f\in\mathcal{F}}$ and $Y$.
	Note that \eqref{eq:distspecial} is invariant with respect to the permutations of $Y_i$.
	This statement will also be true if extended to a joint distribution of \eqref{eq:jointconddist} thanks to Kolmogorov extension theorem (\cite{kolmogorov56}), hence the distribution of joint conditional distribution of \eqref{eq:jointconddist} given $P_n^X$ and $P_n^Y$ is invariant with respect to the permutation of $Y_i$.
\end{proof}

Before proposition \ref{perminvariance} is translated into a result in terms of $R^2$, we formally define $R^2$.
Consider $n$ realizations of $(X, Y)$ and denote them as $(x_1, y_1), ..., (x_n, y_n)$.
Let $L$ be a loss function and $\hat{f}$ be an empirical risk estimator in the sense that
\begin{equation}\label{eq:fhat}
	\hat{f}=\argmin_{f\in\mathcal{F}}\sum_{i=1}^n L(f(x_i\color{black}),y_i\color{black}).
\end{equation}
Let $\hat{f}(x_i)$ denote the prediction of $y_i$ for $i=1,...,n$.
Then
\begin{equation}\label{eq:rsquared}
	R^2=1-\frac{\sum_i(y_i-\hat{f}(x_i))^2}{\sum_i(y_i-\bar{y})^2},
\end{equation}
where $\bar{y}$ is the mean of the $y_i$.
This definition of $R^2$ is the natural one if the loss function $L$ in equation \eqref{eq:fhat} is chosen to be the squared error loss.
In the context of $R^2$, proposition \ref{perminvariance} implies the following result.
\begin{proposition}\label{remarkrsq}
	Let $(X_1,Y_1),...,(X_n,Y_n)$ be an i.i.d. sample of $(X,Y)$.
	Assume that $Y$ is independent of $(f(X))_{f\in\mathcal{F}}$. 
	Fix a permutation $\tau$ of $\{1,...,n\}$ and a loss function $L$ defining an empirical risk estimator as in \eqref{eq:fhat}. Then, conditionally on the empirical measure $P_n^X$ of $X_1,...,X_n$ and the empirical measure $P_n^Y$ of $Y_1,...,Y_n$, the distribution of $R^2$ calculated based on data $\{(X_i,Y_{\tau(i)})\}$ using the aforementioned emprirical risk estimator does not depend on $\tau$.
\end{proposition}
\begin{proof}
	Proposition \ref{perminvariance} implies that the conditional distribution of 
	\begin{equation}\label{eq:jointconddistr2}
		\left(\sum\limits_{i=1}^n(Y_{\tau(i)}-f(X_i))^2,\sum\limits_{i=1}^nL(f(X_i),Y_{\tau(i)})\right)_{f\in \mathcal{F}}
	\end{equation}
	given $P_n^X$ and $P_n^Y$ is the same for all permutations $\tau$ of set $\{1,...,n\}$.
	This is a two-dimensional empirical process indexed by class of functions $\mathcal{F}$.
	Plugging in the $\argmin$ of the second component into the first component still gives a distribution that does not depend on $\tau$.
	Hence, combining the definition \eqref{eq:fhat} of $\hat f$ and \eqref{eq:rsquared} of $R^2$, we conclude that for each permutation $\tau$, $R^2$ calculated for $\{(X_i,Y_{\tau(i)})\}$ is sampled from the same distribution conditioned on $P_n^X$ and $P_n^Y$.
\end{proof}
This allows us to consider $R^2$ as a viable choice for the test statistic.
Under the null hypothesis, the $R^2$ as calculated for $(x_i,y_i)$ is sampled from the same distribution as the $R^2$ calculated for $(x_i,y_{\tau(i)})$ for some permutation $\tau$.
The test itself is based on permutations of the pairings $(x_i,y_i)$.
We reject $H_0$ only if the observed $R^2$ is much larger than "most" of the $R^2$ obtained via random permutations.
Essentially we compare the observed $R^2$ to the distribution of $R^2$ under $H_0$ given specific realizations of $X$ and $Y$, but not their pairings.
It is notable that $R^2$ can also be replaced by some other statistic, as long as it can be calculated using the sample $\{(f(x_i),y_i)\}_i$.
Proposition \ref{perminvariance} permits other statistics to be used instead of $R^2$.
Taking $R^2$ as the test statistic is equivalent to taking empirical risk with respect to quadratic loss as the test statistic.
In that sense, the other tests can also be constructed by considering empirical risks with respect to other losses, e.g. absolute loss or Huber loss.

If the class of functions $\mathcal{F}$ contains the constant functions and the predictors are optimized with respect to the quadratic loss, then $R^2$ calculated for a given $\mathcal{F}$ is always non-negative.
This is true, since given set of observations $\{y_i\}_{i=1,...,N}$, we can always choose $f(X)\equiv\frac{1}{N}\sum\limits_{i=1}^N y_i$ which yields $R^2=0$.
Note that including the constants in $\mathcal{F}$, does not disturb the independence of $Y$ and $(f(X))_{f\in\mathcal{F}}$, since $Y$ is always independent of a set of constant random variables.
While $R^2$ is always non-negative in linear regression models (if the intercept is included), that is not the case for instance in the setting of neural nets.

\renewcommand{\thefootnote}{\fnsymbol{footnote}}
\setcounter{footnote}{0}
\setfnsymbol{wiley}
Given a chosen $\alpha$ level\footnote{default $\alpha=0.05$}, the precise implementation of the test is as follows:
\begin{enumerate}
	\item given original pairings of $(x_i, y_i)$, calculate the $R^2$ of class of functions $\mathcal{F}$, which we will denote as $r_0^2$,\footnote{The specific method of prediction of $\hat{Y}_i$ is stated in \ref{eq:fhat}.}
	\item find the distribution of $R^2$ under the null hypothesis conditionally on observed $x_i$ and $y_{(i)}$ for $i=1,...,n$ (approximated by the empirical distribution function of $R^2$ values based on a uniform sample of permutations of original pairings $(x_i, y_i)$; for each sample $\{(x_i,y_{\tau(i)}):1\leq i\leq n\}$, where $\tau$ is a permutation, $R^2$ is calculated; notably, the model is refit for each permuted sample),
	\item if $r_0^2>q_{1-\alpha}$, where $q_{1-\alpha}$ is the $1-\alpha$ quantile of the empirical distribution of $R^2$ values, then we reject the null hypothesis, otherwise we do not reject it.
\end{enumerate}
Any tuning parameters used in point (1) and (2) are not adjusted for each permutation.
This implementation assumes that $R^2$ is the statistic of choice, but it can be adapted to suit other statistics as well.
The reason we prioritize $R^2$ is primarily because of its benefits in terms of interpretability and ease of use.
It is also important to note that in practice, determining the distribution of $R^2$ under the null hypothesis will not be exact in most cases.
To obtain the exact distribution we need to run through $n!$ permutations.
Even for $n>10$ the computational cost of such an operation is prohibitively expensive and sampling from the true distribution is more reasonable.

\begin{figure}[hbtp!]
	\begin{center}
		\includegraphics[width=\textwidth]{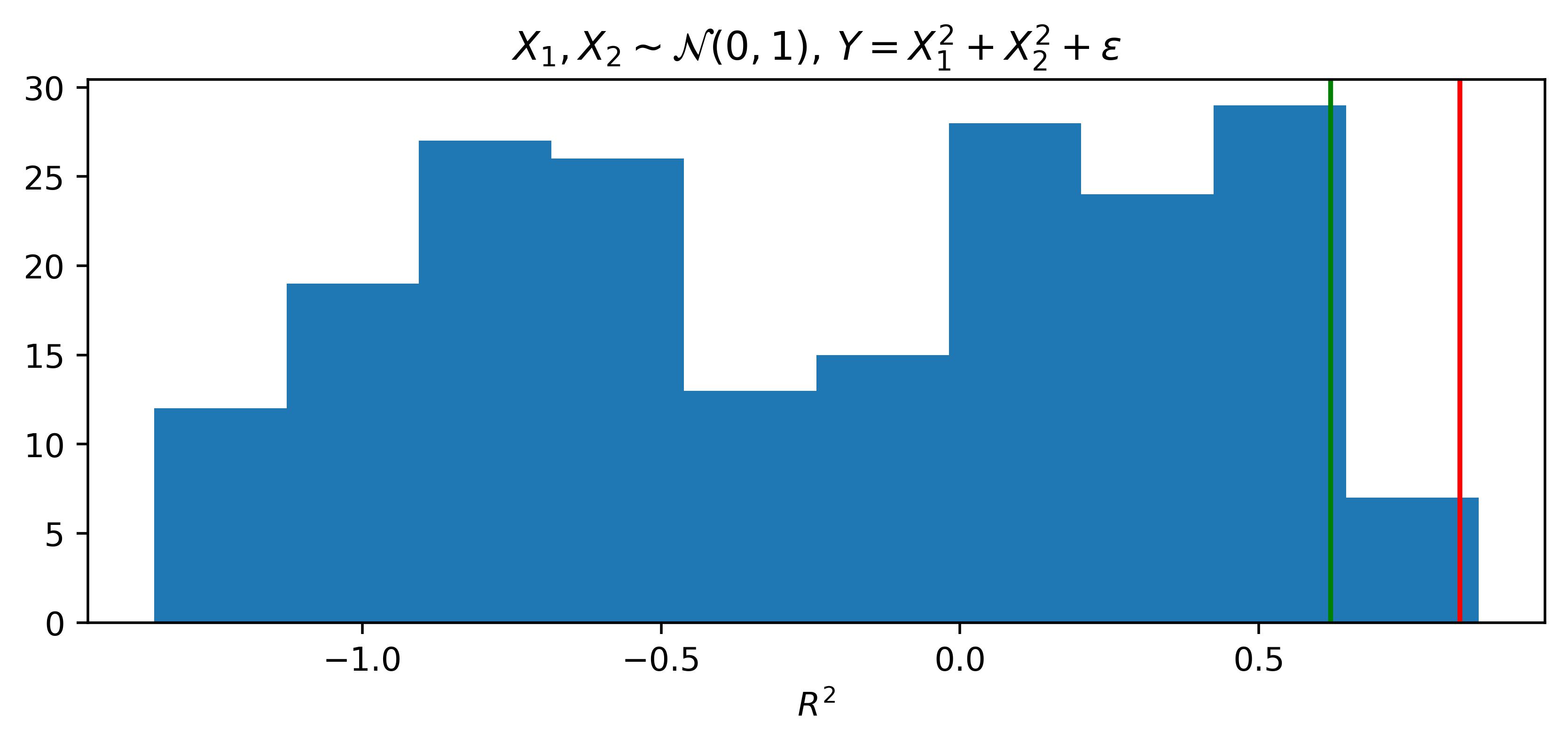}
		\caption{Histogram of the distribution of generated $R^2$ using permutation of $y$ values. The model considered here is a 3-layered neural net. The sample size is 10. The red line denotes the observed $R^2$ for the true pairings of $X$ and $Y$, the green line denotes the 95\%-quantile of the empirical distribution of $R^2$ (approximation using 200 permutations).}
		\label{fig:example}
	\end{center}
\end{figure}

$R^2$ is bounded by 1 from above for any model.
The proximity of $R^2$ values calculated from the permuted data or original $R^2$ values to 1 or to each other can provide insight into goodness of fit of a model.
The closer the values of $R^2$ for the permuted data to 1, the greater the capability of the model to fit to the noise.
Close proximity of $q_{1-\alpha}$ to $r_0^2$ in case of $r_0^2>q_{1-\alpha}$ and $r_0^2$ small implies that the model's predictive ability may not be satisfactory even though the null hypothesis is rejected by the test.
The test is widely applicable, because of its general form and easily adaptable to different types of models.
It also provides an interesting commentary on the predictive abilities of a chosen model.
In the event that the quantile $q_{1-\alpha}$ for one model, $\mathcal{F}_1$, significantly exceeds the same quantile for another model, $\mathcal{F}_2$, we conclude that $\mathcal{F}_1$ is either overfitting, indicating a need for reduction of the set of independent variables or model simplification, or is better able to extract meaningful information from unrelated data.\newline
\indent We close this section with a continuation of the example from section \ref{sec:problem}.
Let $X_1, X_2$ be independent standard normal variables and $Y=X_1^2+X_2^2+\epsilon$, where $\epsilon\sim\mathcal{N}(0,0.01)$.
We consider a neural net as a model of choice to predict $Y$ using $X_1$ and $X_2$.
A random sample of size 10 is drawn.
We conduct the permutation test.
As seen in figure \ref{fig:example}, the test rejects the null hypothesis.
However, in the case of 2 out of 200 permutations the model achieves higher $R^2$ than in the case of the original pairings.
Even though the model is capable to capture the relationship between $X_1$, $X_2$ and $Y$, there are permutations of the vector of responses that can lead to a better performance of the model.

\section{Application}

\subsection{Tennis serve dataset}\label{sec:application}

This section concerns an application of the permutation test to a tennis serve dataset.
Seven professional athletes wearing inertial measurement units (IMUs) performed tennis serves.
Each athlete followed a protocol of first and second serves.
Sensors were placed on 4 body parts: lower and upper arms, trunk and pelvis as can be seen in fig. \ref{fig:segments}.
Each IMU contained a triaxial accelerometer and triaxial gyroscope.
The data consists of 7 uninterrupted time series of 24-dimensional data (4 body parts $\times$ 2 types of sensors $\times$ 3 axes).
The dataset is further described in the Master thesis (\cite{faneker21}).

\begin{figure}[hbtp!]
	\begin{center}
		\includegraphics[width=0.55\textwidth]{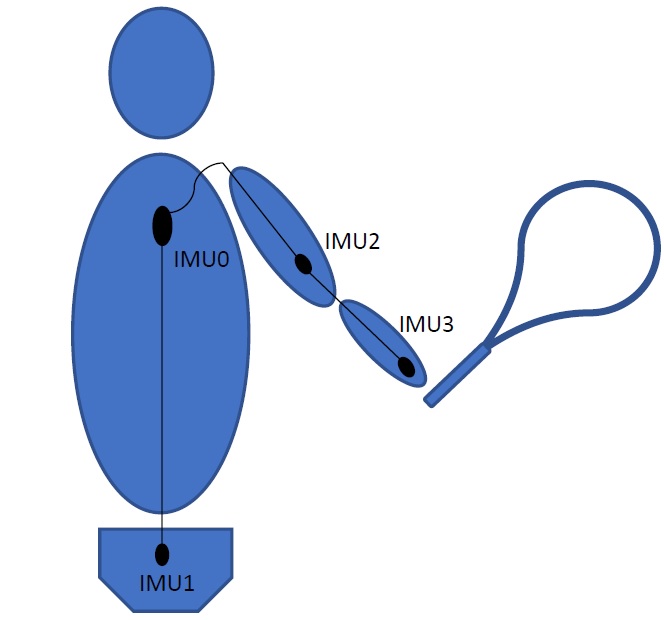}
		\caption{Segment model of right-handed player and racquet (back view, frontal plane).}
		\label{fig:segments}
	\end{center}
\end{figure}

Additionally, a dataset containing personal characteristics of the players and performance characteristics of each serve has been included.
The personal characteristics are the sex, age, height and weight of the players.
The performance characteristics are the ball velocity, an indication of whether the ball went in or out and the velocity-accuracy index (VA index).
The VA index for a single serve was introduced and motivated by \cite{kolman17} and is defined as follows:
\begin{equation}\label{eq:vaindex}
	\textrm{VA index}=\frac{(\textrm{ball velocity (kph)})^2}{100}\times\frac{\textrm{achieved points}}{9},
\end{equation}
where achieved points refer to the number of points assigned to a serve based on its closeness to a target area on the court (see fig. \ref{fig:vapoints}).
The number of points assigned to a serve is based on a new Serve Tennis Test (STT) adapted from \cite{kolman17}.
Originally, the point system was devised based on the ellipses in the serve box where aces were hit in male tennis matches during the Australian Open (\cite{whiteside17}).
However, the system has been improved upon since then.
The points are discrete.
Nine points are given for hitting the center of the target area. 
Six and three points are assigned for areas further from the center.
One point is assigned for a ball much further from the target area, but still a valid ball, while zero points are given to a serve which did go out.
Each participant performed approximately 48 serves.
In total, 29.6\% of serves were faults (and as a result had a VA-index 0).

\begin{figure}[hbtp!]
	\begin{center}
		\includegraphics[width=0.7\textwidth]{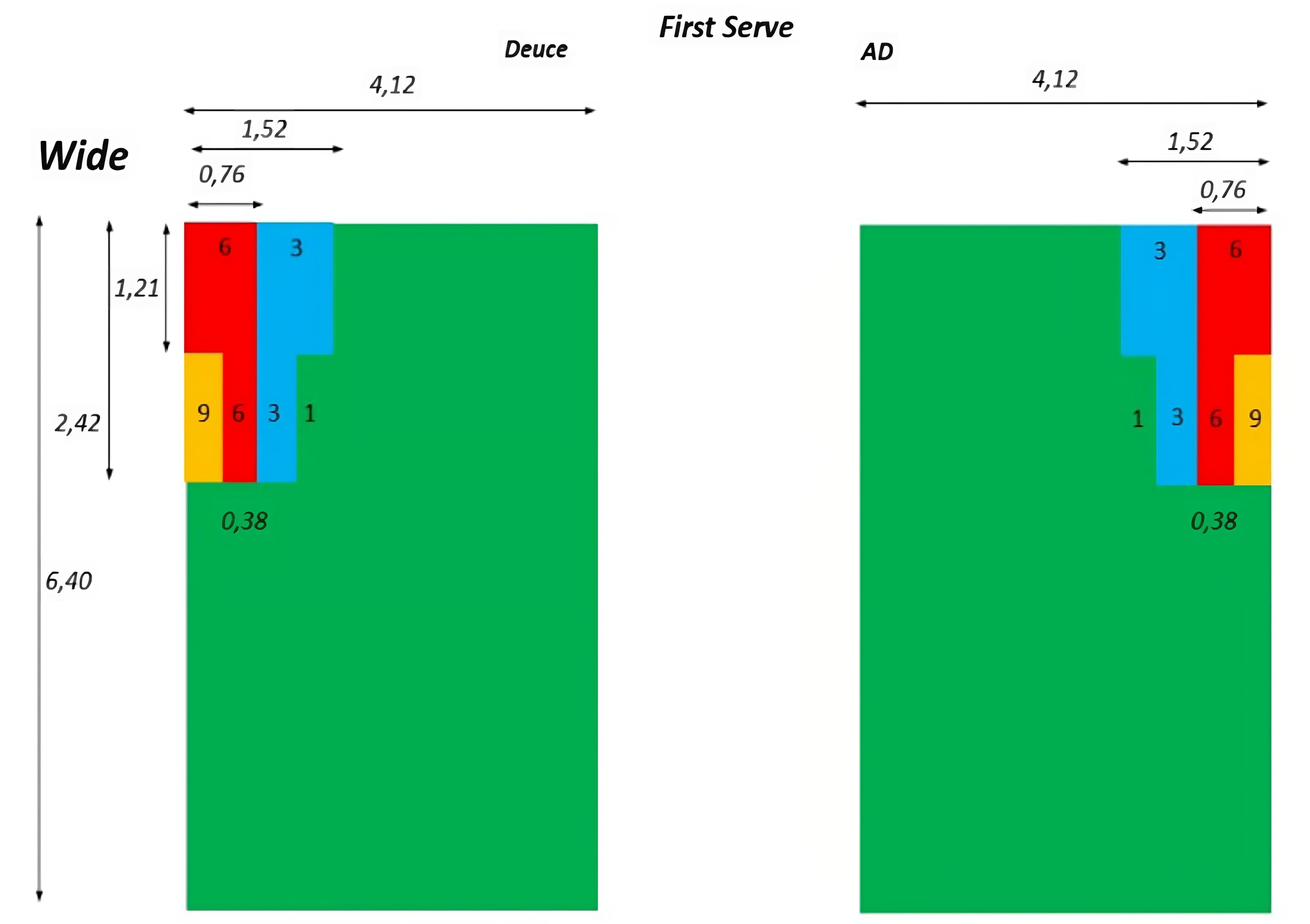}
		\caption{Target areas for the tennis serve. The scenario considered here is a serve in the wide direction. The points given on each target area correspond to the number of accuracy points needed to calculate the VA index of the serve.}
		\label{fig:vapoints}
	\end{center}
\end{figure}

We will use the tennis serve dataset in order to demonstrate an application of the permutation test to real life data.
We will focus on the prediction of ball speed and VA-index prediction.
The functional predictors have been transformed into vectors, using a Fourier basis representation, in order to be able to use the linear regression model with the class of functions $\mathcal{F}_{\textrm{LR}}$ and the neural net with the class of functions $\mathcal{F}_{\textrm{NN}}(300, 300, 300)$.
The choice to use Fourier coefficients as predictors was the most natural way of incorporating information from the time series.
First, a prediction of ball speed was considered.
The permutation test rejected the null hypothesis in cases of both models as seen in fig. \ref{fig:ballspeed_lr} and \ref{fig:ballspeed_nn}.
The test rejects the null hypothesis for both models, although higher values of $R^2$ achieved by the neural net for the original pairings suggest greater capabilities of that model to detect the dependence.

\begin{figure}
	\centering
	\begin{subfigure}[b]{\textwidth}
		\centering
		\includegraphics[width=\textwidth]{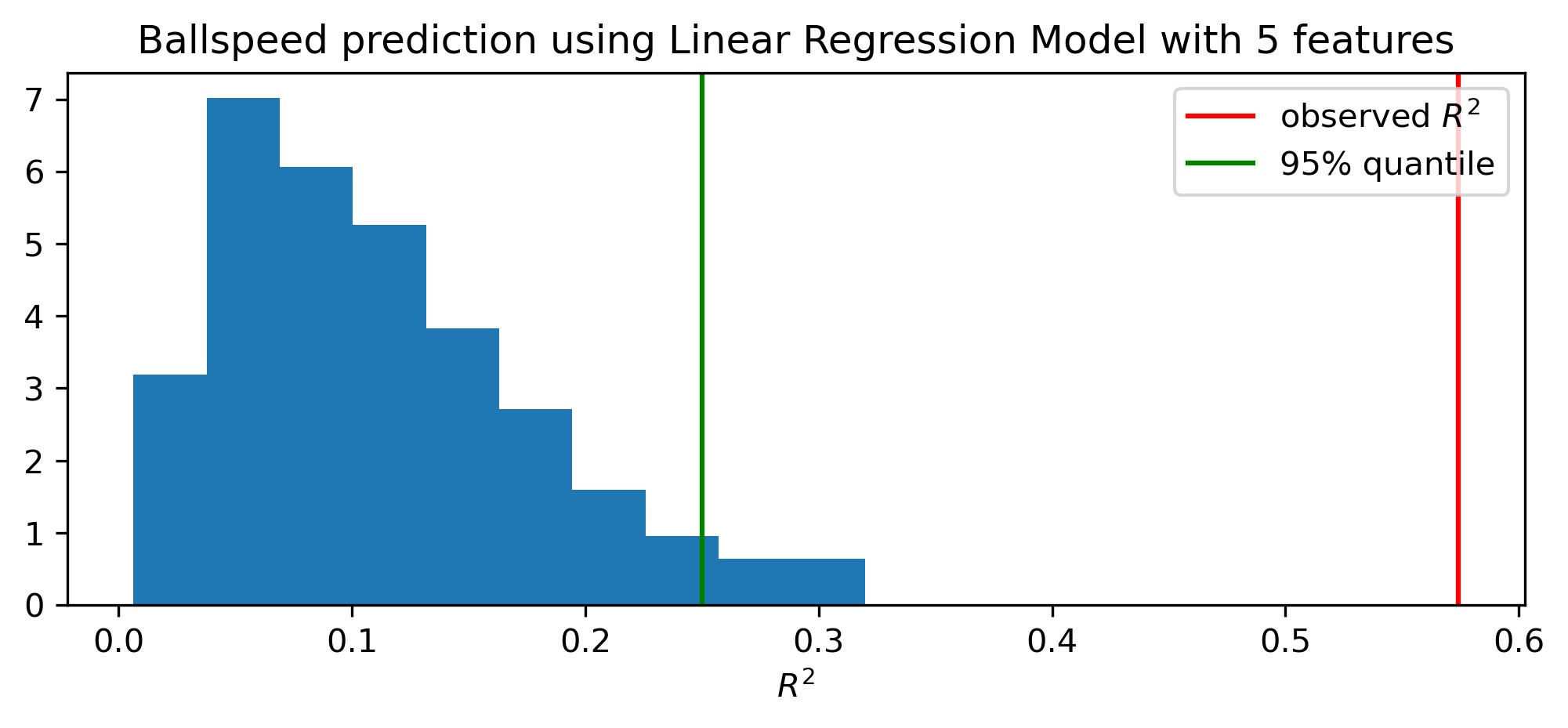}
		\caption{Histogram of the distribution of generated $R^2$ using permutation of $y$ values. The sample size is 46. The red line denotes the observed $R^2$ for the true pairings of $X$ and $Y$, the green line denotes the 95\%-quantile of the empirical distribution of $R^2$ (approximation using 200 permutations).}
		\label{fig:ballspeed_lr}
	\end{subfigure}
	\hfill
	\begin{subfigure}[b]{\textwidth}
		\centering
		\includegraphics[width=\textwidth]{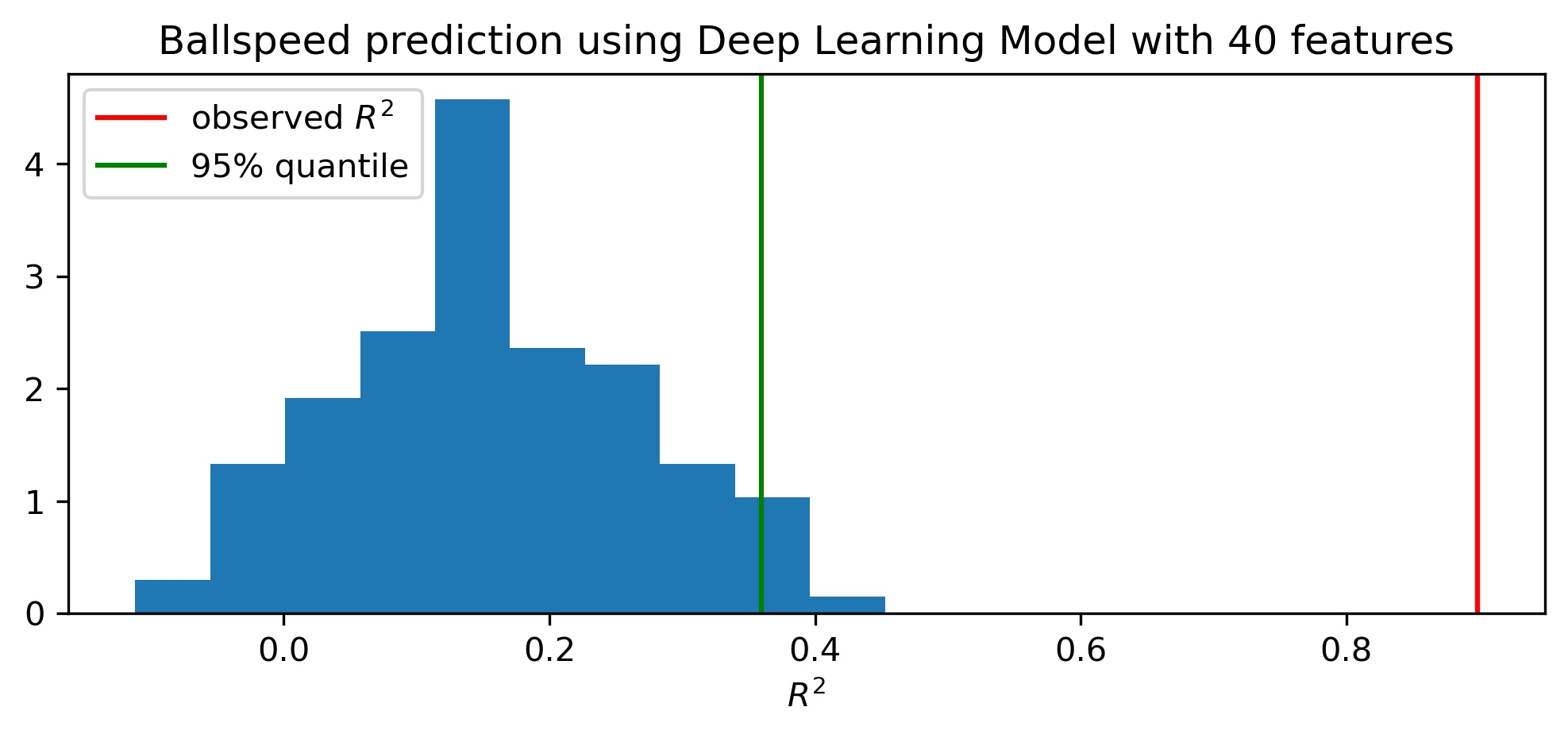}
		\caption{Histogram of the distribution of generated $R^2$ using permutation of $y$ values. The sample size is 46. The red line denotes the observed $R^2$ for the true pairings of $X$ and $Y$, the green line denotes the 95\%-quantile of the empirical distribution of $R^2$ (approximation using 200 permutations).}
		\label{fig:ballspeed_nn}
	\end{subfigure}
	\caption{Results of the permutation test for the ball speed prediction using $\mathcal{F}_{\textrm{LR}}$ and $\mathcal{F}_{\textrm{NN}}(300, 300, 300)$.}
\end{figure}

In the case of prediction of the VA-index as defined in \eqref{eq:vaindex}, the permutation test did not reject the null hypothesis for the linear regression model with the class of functions $\mathcal{F}_{\textrm{LR}}$ as well as for the neural net model with the class of functions $\mathcal{F}_{\textrm{NN}}(300, 300, 300)$.
Fig. \ref{fig:vaindex_lr} shows results for the linear regression model and fig. \ref{fig:vaindex_nn} shows results for the neural net.
The values of $R^2$ are quite low for both models and for many permutations of $y$-values the generated $R^2$ is much higher than the observed $R^2$ for the true pairings.
These results convince us that a good prediction using the linear regression model or the neural net model is not possible at the moment.
The issue may lie with the current size of the dataset or the number of serves per player or simply because the relation as can be described by the neural net is not strong.
The fact that the number of Fourier coefficients used in this prediction was increased to achieve more favourable $R^2$ for the original pairings of $(x_i,y_i)$ (at least in the case of the deep learning model), shows how complex this task is and additional information is needed in the data to increase the $R^2$.

\begin{figure}[hbtp!]
	\centering
	\begin{subfigure}[b]{\textwidth}
		\centering
		\includegraphics[width=0.95\textwidth]{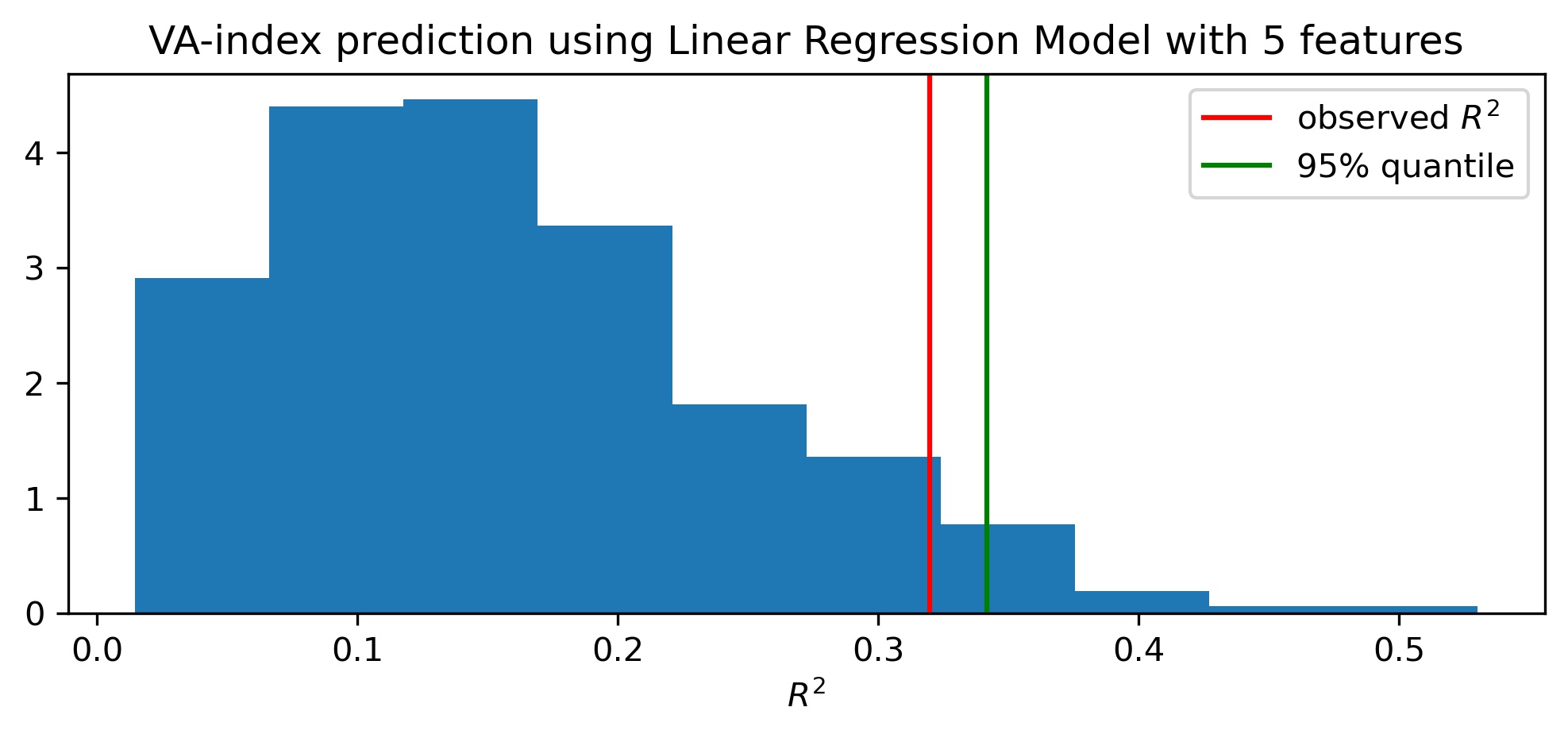}
		\caption{Histogram of the distribution of generated $R^2$ using permutation of $y$ values. The sample size is 34. The red line denotes the observed $R^2$ for the true pairings of $X$ and $Y$, the green line denotes the 95\%-quantile of the empirical distribution of $R^2$ (approximation using 200 permutations).}
		\label{fig:vaindex_lr}
	\end{subfigure}
	\hfill
	\begin{subfigure}[b]{\textwidth}
		\centering
		\includegraphics[width=0.95\textwidth]{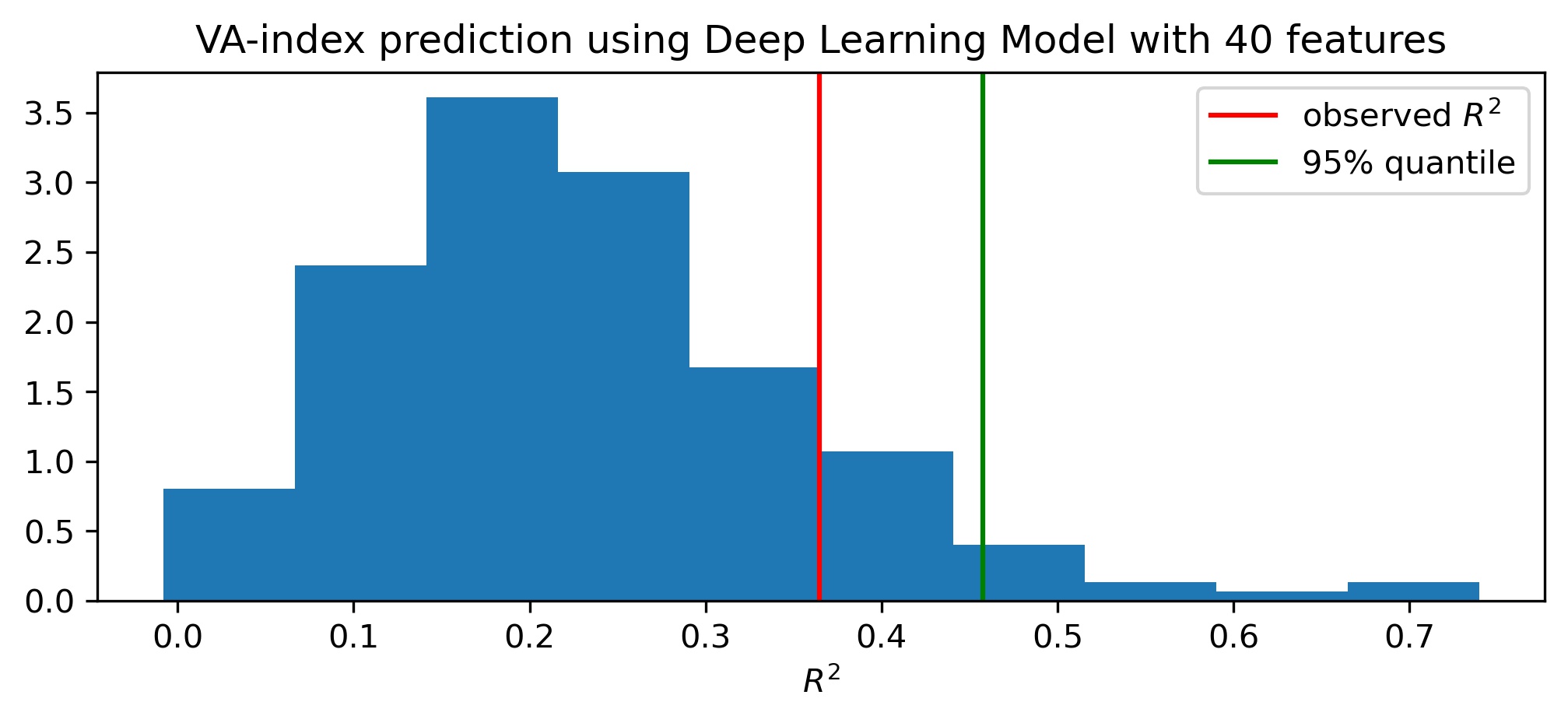}
		\caption{Histogram of the distribution of generated $R^2$ using permutation of $y$ values. The sample size is 34. The red line denotes the observed $R^2$ for the true pairings of $X$ and $Y$, the green line denotes the 95\%-quantile of the empirical distribution of $R^2$ (approximation using 200 permutations).}
		\label{fig:vaindex_nn}
	\end{subfigure}
	\caption{Results of the permutation test for the VA index prediction using $\mathcal{F}_{\textrm{LR}}$ and $\mathcal{F}_{\textrm{NN}}(300, 300, 300)$.}
\end{figure}

\newpage

\section{Conclusion and discussion}\label{sec:conclusions}

This paper concerns the theoretical foundations and the application of the permutation approach for testing whether a model can capture dependence structure between predictors and responses.
The test is a tool to determine whether a model is able to fit the data better than pure noise.
We are mostly interested whether $X$ has any effect on $Y$ and we pursue that interest with the help of a chosen, fixed model.
The null hypothesis is formulated in terms of independence of $Y$ and $(f(X))_{f\in\mathcal{F}}$ and in this form cannot be found in previous literature.
Proposition \ref{perminvariance} allows us to consider the test as a permutation test formally and proposition \ref{remarkrsq} allows us to consider $R^2$ as a test statistic.
This approach is data-centered and the results of the test depend on just one model without the need to directly compare between different models.
We also do not require sample splitting thus the test can rely on the power of the whole sample size, which can be vital in datasets of smaller size.
Our findings are supported through an application to the tennis serve dataset. 
In this case, it gave evidence that a seemingly well-fitting model is not necessarily trustworthy.
The prediction is either not possible with the given sensor data and model or a larger sample size is needed to predict the VA-index more accurately.

\subsection*{Declaration of interests}
The authors declare that they have no known competing financial interests or personal relationships that could have appeared to influence the work reported in this paper.

\subsection*{Code availability}
Custom code for simulation study can be found at: \url{https://github.com/mgciszewski/credibility_2023}.

\subsection*{Data availability}
The data that support the findings of this study are available from the corresponding author upon request.

\subsection*{Acknowledgments}
We thank the two anonymous referees for valuable comments to the manuscript, in particular on the formulation of an earlier version of the null hypothesis.

\bibliography{references}

\appendix

\section{Simulation study} \label{sec:simstudy}

We apply our permutation test in multiple scenarios.
This section will specifically focus on simulated datasets to assess the test's performance on datasets with varying dependence levels between $X$ and $Y$ and two different class of functions $\mathcal{F}$.
An empirical example will be considered in section \ref{sec:application}.
In all scenarios we consider the $R^2$-based test.

Two different models will be used to fit the data throughout this section.
One of them is a linear regression model, which models the relationship between a random vector $X$ and a random variable $Y$ in a linear manner: $Y=\beta\cdot X + \epsilon$.
The parameter vector $\beta$ will always be estimated using the least squares method.
Regardless of the length of vector $X$, the class of functions associated with this model will be referred to as $\mathcal{F}_{\textrm{LR}}$.
The other model we consider is a neural net.
A neural net is a collection of neurons arranged into layers, with neurons from different layers connected to each other.
Typically, a neural net consists of an input layer, multiple hidden layers and an output layer.
The estimation of neural nets' parameters, the weights associated with neurons and edges between them, is done by feeding multiple training sets of inputs and outputs into the net.
Weights are adjusted each time based on a predefined cost function.
Class of functions associated with neural nets will be referred to as $\mathcal{F}_{\textrm{NN}}$ with the number of neurons on each layer specified as a $k$-tuple, where $k$ refers to the number of layers, e.g. $\mathcal{F}_{\textrm{NN}}(30,30,30)$ is a neural net with 3 hidden layers, each of which contains 30 neurons.

In the first two examples, we will compare the permutation test to two existing methods: Spearman's rank correlation coefficient (also referred to as Spearman's $\rho$) and Kendall rank correlation coefficient (also referred to as Kendall's $\tau$).
Both are statistics used to measure the rank correlation between two variables and both can be used as test statistics in a test for independence of two variables.
Since, our examples have more than one explanatory variable, multiple statistics will be given.
It is worth noting that both statistics are not applicable when there is no natural ordering in the data, e.g. in the case of functional data when datapoints are functions.

Let $X_1,X_2\sim\mathcal{N}(0,1)$ and $Y\sim U([0,1])$ be independent random variables.
We consider two models and two classes of functions associated with them: $\mathcal{F}_{\textrm{LR}}$ and $\mathcal{F}_{\textrm{NN}}(30,30,30)$ and a sample of size $100$.
In both cases the null hypothesis is not rejected, see fig. \ref{fig:norelation_lin_reg_one} and \ref{fig:norelation_nn_one}.
We also consider 1000 repetitions of the experiment in the same setup to see the behavior of the test on a larger number of examples.
As seen in fig. \ref{fig:norelation_lin_reg_repeats} and \ref{fig:norelation_nn_repeats}, the null hypothesis is rejected in most repetitions for both models, namely 4.7\% for the linear model and 4.5\% for the neural net.
This shows that the rejection of the null hypothesis can still happen even in case of independence.
Most importantly, the rejection rate is close to the confidence level $\alpha=5\%$.
Spearman's $\rho$ test rejects the null hypothesis of independence of $X_1$ and $Y$ in $5.2\%$ of all cases and rejects the independence of $X_2$ and $Y$ in $5\%$ of all cases.
Kendall's $\tau$ test rejects the null hypothesis of independence of $X_1$ and $Y$ in $5.2\%$ of all cases and rejects the independence of $X_2$ and $Y$ in $5.3\%$ of all cases.
For both of these tests, the rejection rate is also close to the confidence level.

\begin{figure}
	\centering
	\begin{subfigure}[b]{\textwidth}
		\centering
		\includegraphics[width=\textwidth]{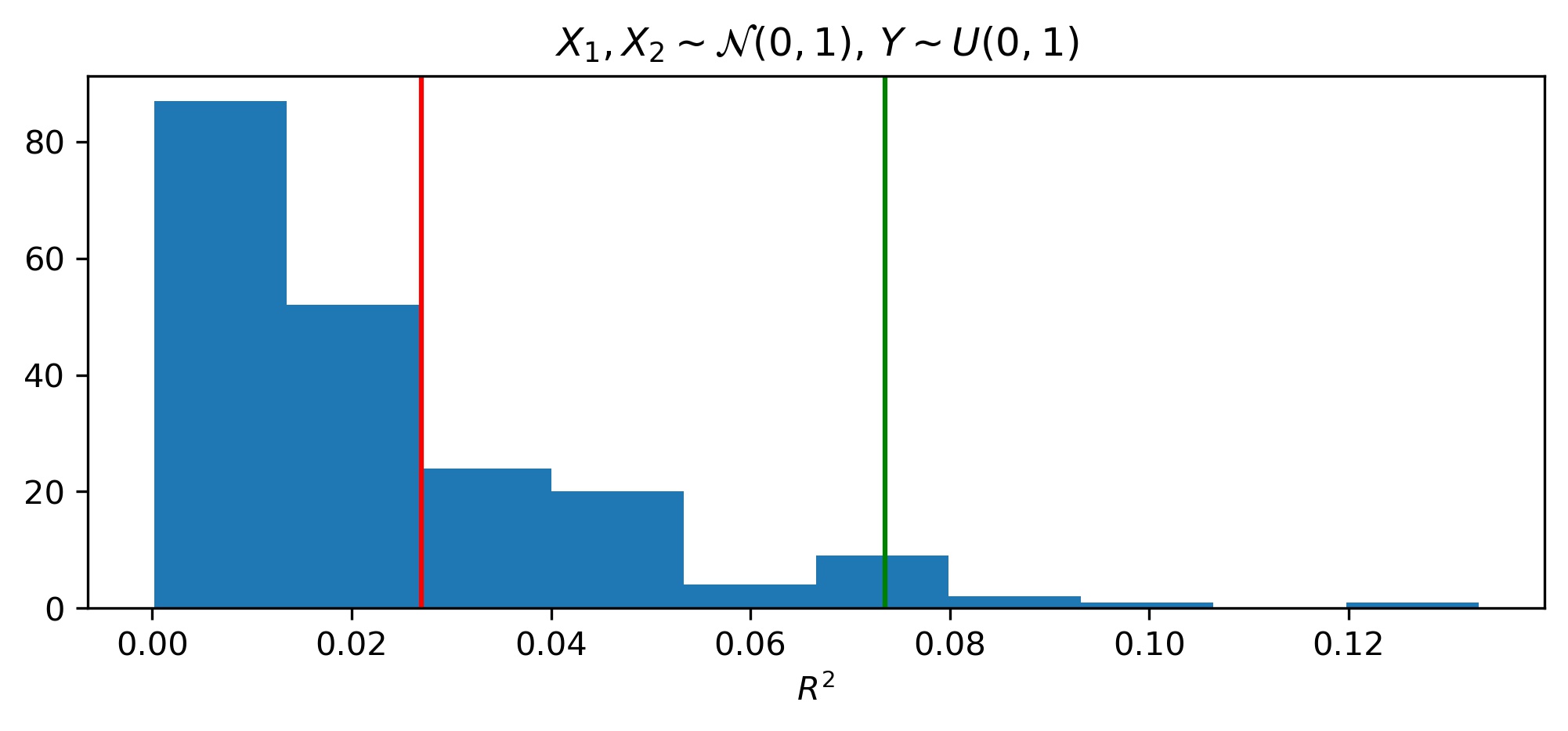}
		\caption{Histogram of the distribution of generated $R^2$ using permutation of $y$ values. The model considered here is linear regression with the class of functions $\mathcal{F}_{\textrm{LR}}$. The sample size is 100. The red line denotes the observed $R^2$ for the true pairings of $X$ and $Y$, the green line denotes the 95\%-quantile of the empirical distribution of $R^2$ (approximation using 200 permutations).}
		\label{fig:norelation_lin_reg_one}
	\end{subfigure}
	\hfill
	\begin{subfigure}[b]{\textwidth}
		\centering
		\includegraphics[width=\textwidth]{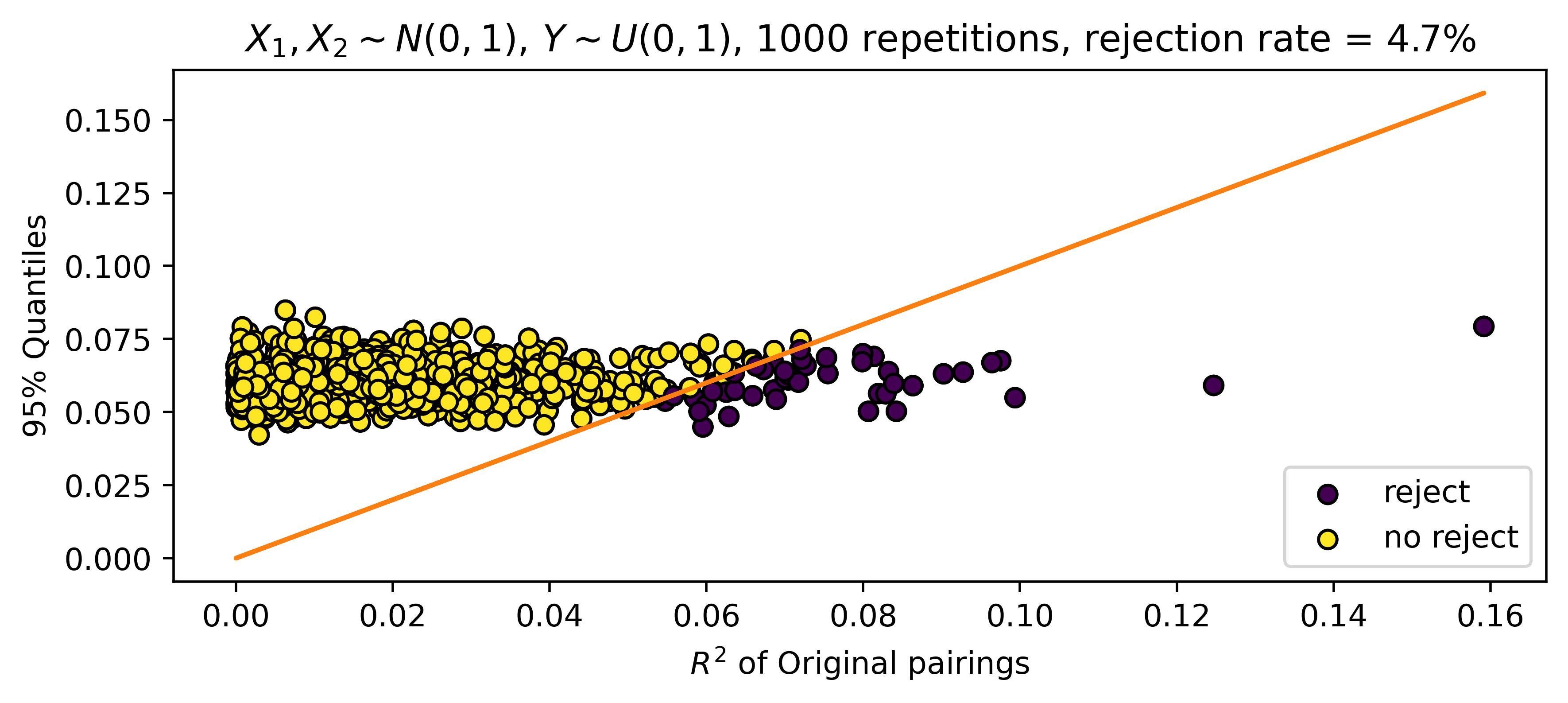}
		\caption{Scatterplot of the $R^2$ values for the original pairings against the 95\% quantiles of the empirical distribution of $R^2$. The orange line shows the identity function.}
		\label{fig:norelation_lin_reg_repeats}
	\end{subfigure}
	\caption{Results of the permutation test for $\mathcal{F}_{\textrm{LR}}$ with data generated in a following manner $X_1,X_2\sim\mathcal{N}(0,1)$ and $Y\sim U([0,1])$.}
	\label{fig:norelation_lin_reg}
\end{figure}

\begin{figure}
	\centering
	\begin{subfigure}[b]{\textwidth}
		\centering
		\includegraphics[width=\textwidth]{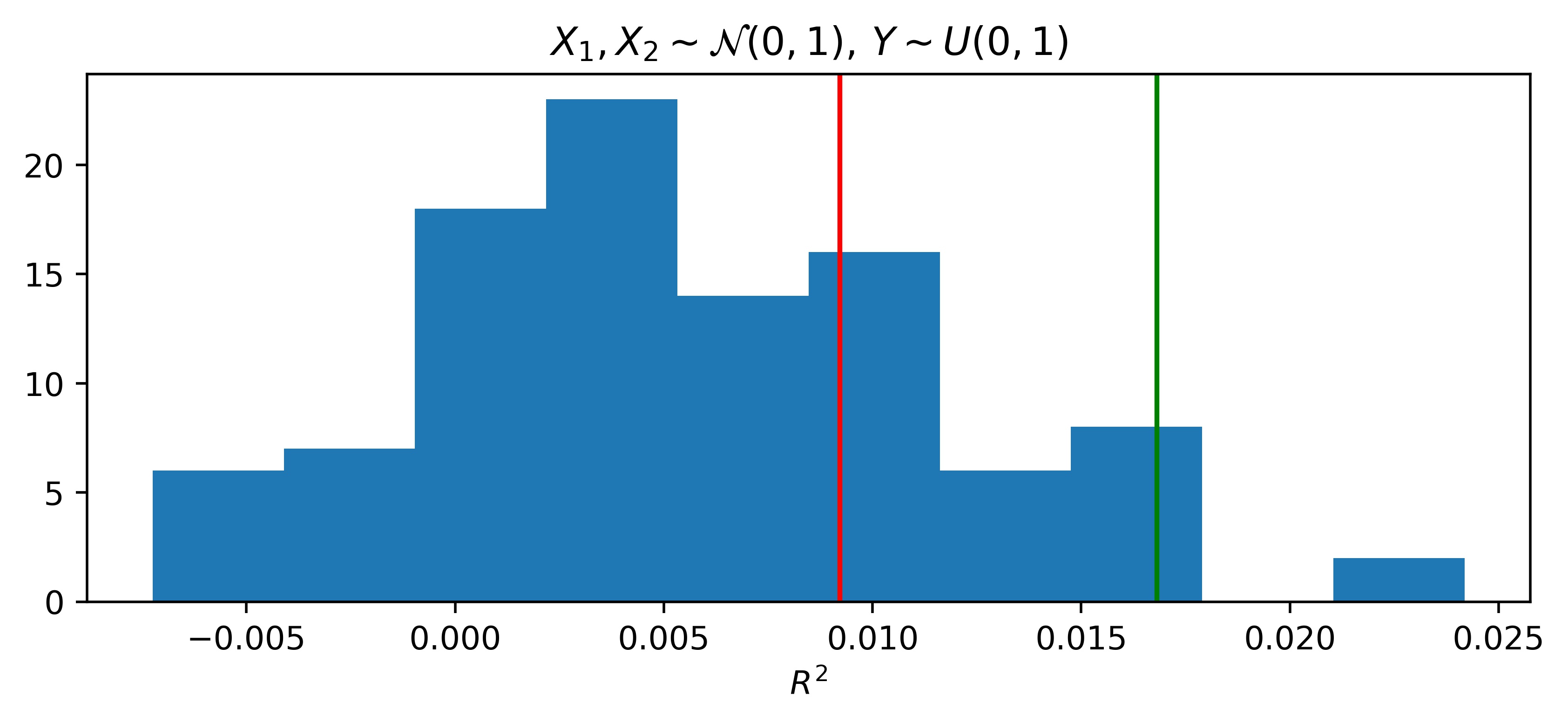}
		\caption{Histogram of the distribution of generated $R^2$ using permutation of $y$ values. The model considered here is a 3-layered neural net with the class of functions $\mathcal{F}_{\textrm{NN}}(30,30,30)$. The sample size is 100. The red line denotes the observed $R^2$ for the true pairings of $X$ and $Y$, the green line denotes the 95\%-quantile of the empirical distribution of $R^2$ (approximation using 200 permutations).}
		\label{fig:norelation_nn_one}
	\end{subfigure}
	\hfill
	\begin{subfigure}[b]{\textwidth}
		\centering
		\includegraphics[width=\textwidth]{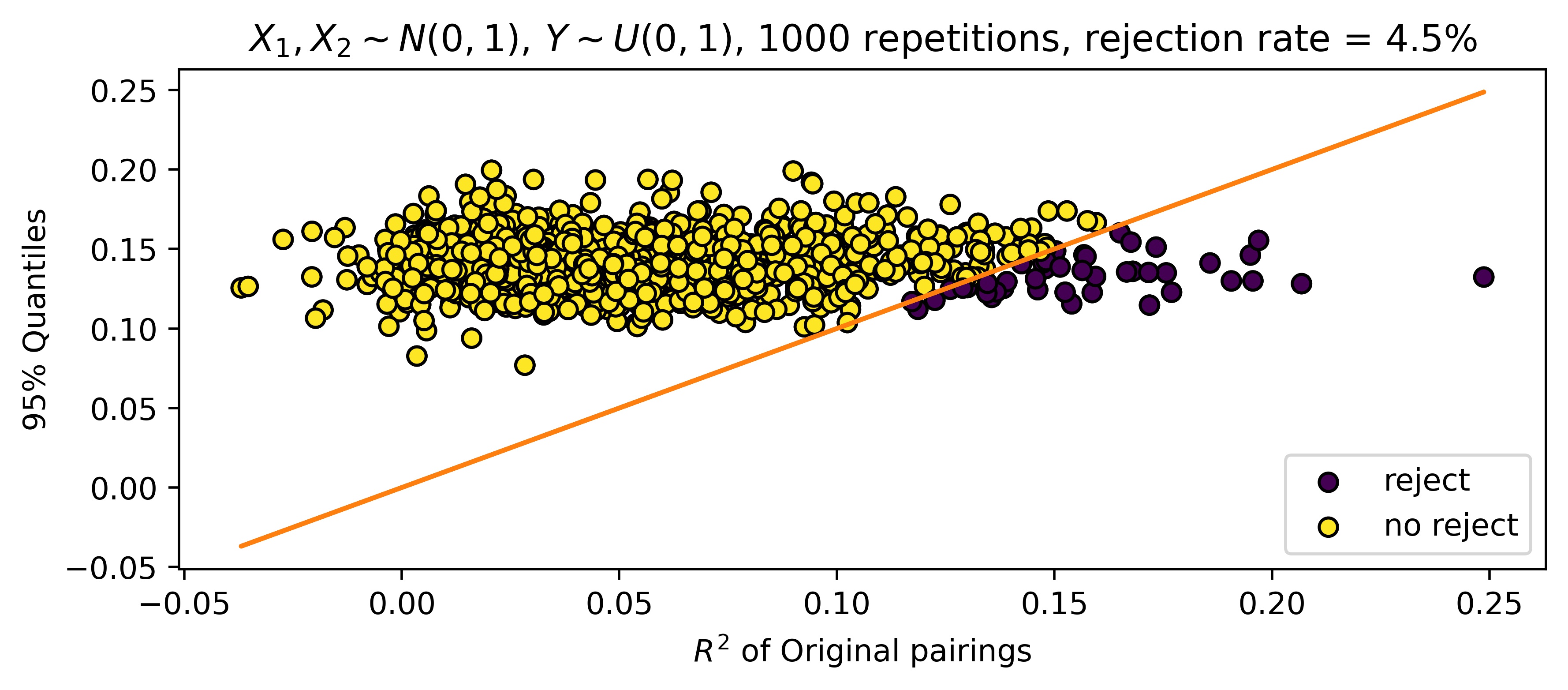}
		\caption{Scatterplot of the $R^2$ values for the original pairings against the 95\% quantiles of the empirical distribution of $R^2$. The orange line shows the identity function.}
		\label{fig:norelation_nn_repeats}
	\end{subfigure}
	\caption{Results of the permutation test for $\mathcal{F}_{\textrm{NN}(30, 30, 30)}$ with data generated in a following manner $X_1,X_2\sim\mathcal{N}(0,1)$ and $Y\sim U([0,1])$.}
	\label{fig:norelation_nn}
\end{figure}

Now, let $X_1\sim\mathcal{N}(1,1),X_2\sim\mathcal{N}(0,1)$ be independent and $Y=\log|X_1|+X_2^2+\epsilon$, where $\epsilon\sim\mathcal{N}(0,1)$ is the noise.
Consider a sample of size $100$.
For both $\mathcal{F}_{\textrm{LR}}$ and $\mathcal{F}_{\textrm{NN}}(30,30,30)$, the permutation test rejects the null hypothesis, since the values of $R^2$ for the original pairings are much higher than for any of the permuted pairings.
For the behavior of the test in a single example see fig. \ref{fig:relation_lin_reg_one} and \ref{fig:relation_nn_one}.
In this case the neural net outperforms the linear model significantly, thanks to its complexity.
Fig. \ref{fig:relation_lin_reg_repeats} and \ref{fig:relation_nn_repeats} show that the rejection rate in this case is quite high when repeating the experiment 1000 times, close to 95\% for the linear model and 94\% for the neural net.
This particular example illustrates the test's applicability in the case of a functional relation between predictors and responses.
The model is not just fitting the noise, there is some relation between predictors and responses.
It might not be captured well using a linear regression model, but the model is still able to capture more than pure noise.
Spearman's $\rho$ test and Kendall's $\tau$ test have also been performed in this example, but they show a slight difference from what we see in the case of the permutation test.
Spearman's $\rho$ test rejects the null hypothesis of independence of $X_1$ and $Y$ in $99.6\%$ of all cases and rejects the independence of $X_2$ and $Y$ in $9.5\%$ of all cases.
Similarly, Kendall's $\tau$ test rejects the null hypothesis of independence of $X_1$ and $Y$ in $99.7\%$ of all cases and rejects the independence of $X_2$ and $Y$ in $11.7\%$ of all cases.
This shows that the relationship between $X_1$ and $Y$ is easier to capture than the relationship between $X_2$ and $Y$, and with high probability the test will indicate that $X_1$ and $Y$ are not independent.
The relationship between $X_2$ and $Y$ is not as easy to capture using Kendall's $\tau$ or Spearman's $\rho$, which is to be expected due to the application of a nonlinear function with a minimum at the mean of $X_2$ when defining $Y$.

\begin{figure}
	\centering
	\begin{subfigure}[b]{\textwidth}
		\centering
		\includegraphics[width=\textwidth]{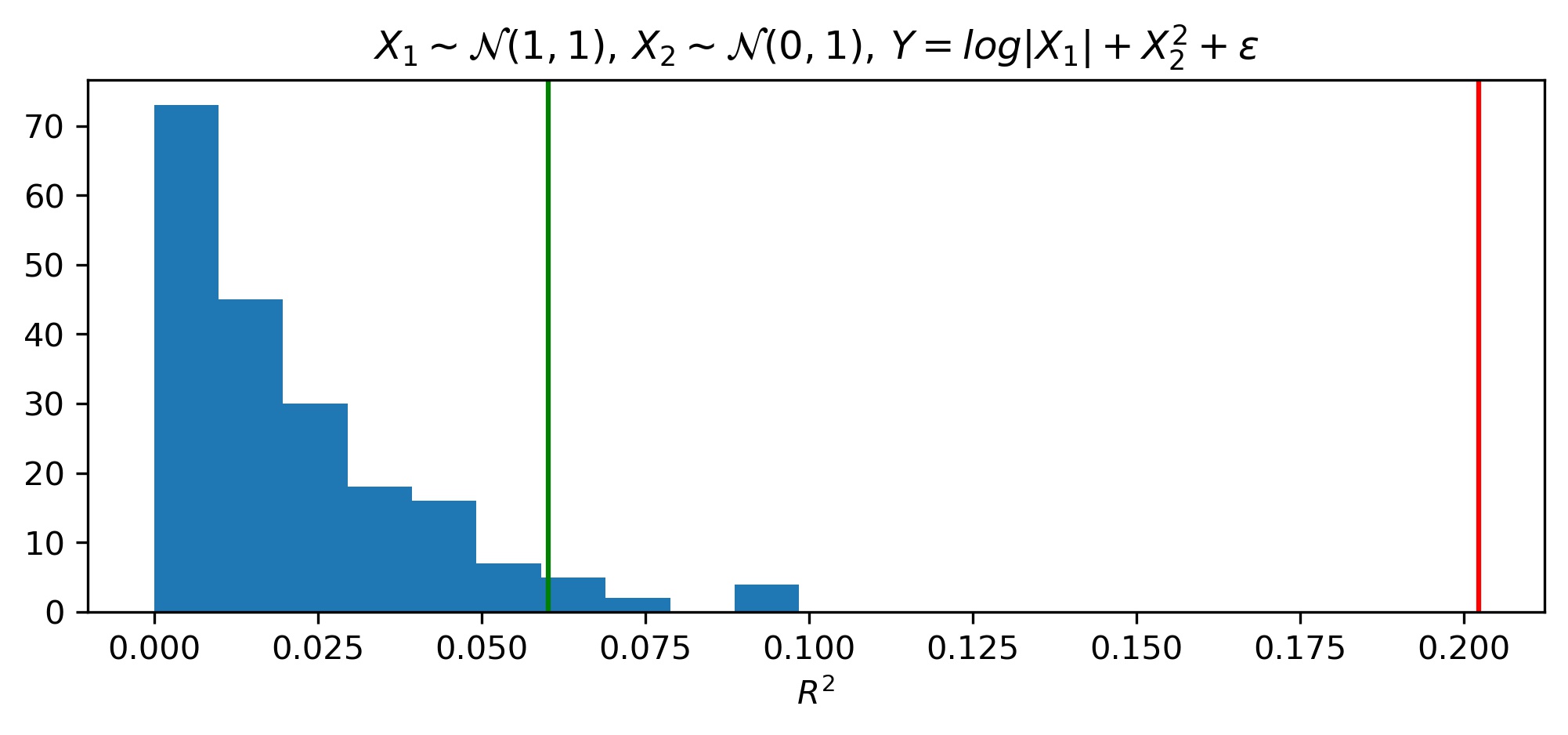}
		\caption{Histogram of the distribution of generated $R^2$ using permutation of $y$ values. The model considered here is linear regression with the class of functions $\mathcal{F}_{\textrm{LR}}$. The sample size is 100. The red line denotes the observed $R^2$ for the true pairings of $X$ and $Y$, the green line denotes the 95\%-quantile of the empirical distribution of $R^2$ (approximation using 200 permutations).}
		\label{fig:relation_lin_reg_one}
	\end{subfigure}
	\hfill
	\begin{subfigure}[b]{\textwidth}
		\centering
		\includegraphics[width=\textwidth]{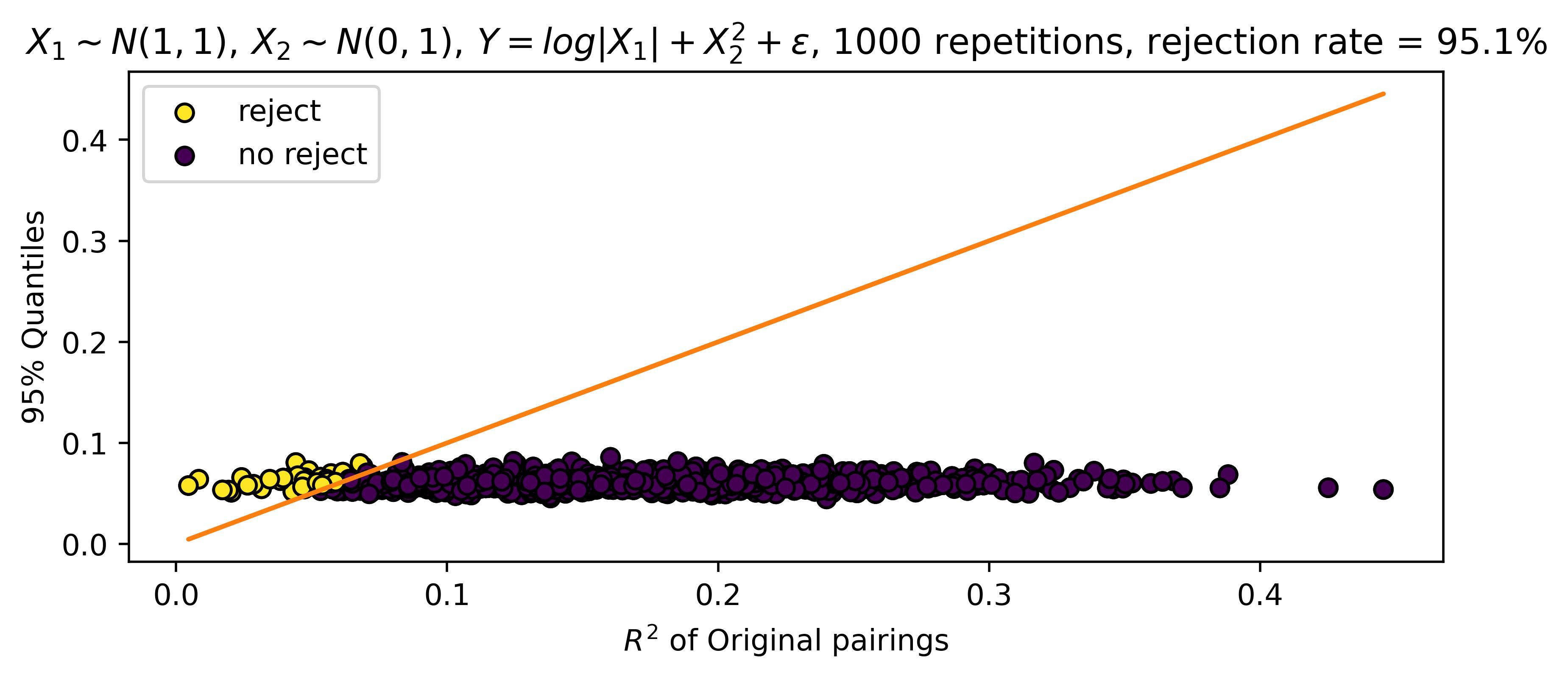}
		\caption{Scatterplot of the $R^2$ values for the original pairings against the 95\% quantiles of the empirical distribution of $R^2$. The orange line shows the identity function.}
		\label{fig:relation_lin_reg_repeats}
	\end{subfigure}
	\caption{Results of the permutation test for $\mathcal{F}_{\textrm{LR}}$ with data generated in a following manner $X_1\sim\mathcal{N}(1,1),X_2\sim\mathcal{N}(0,1)$ and $Y=\log|X_1|+X_2^2+\epsilon$, where $\epsilon\sim\mathcal{N}(0,1)$.}
	\label{fig:relation_lin_reg}
\end{figure}

\begin{figure}
	\centering
	\begin{subfigure}[b]{\textwidth}
		\centering
		\includegraphics[width=\textwidth]{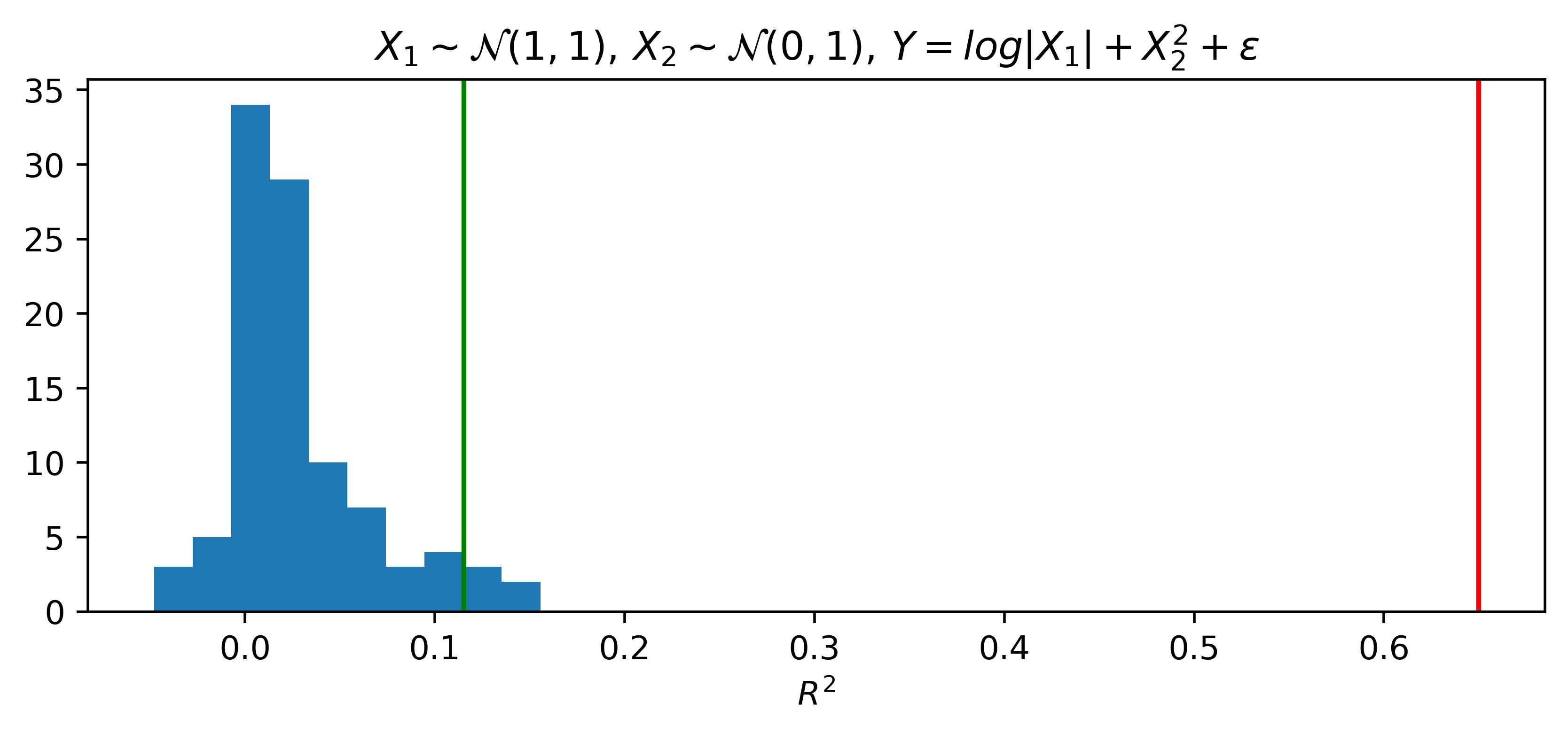}
		\caption{Histogram of the distribution of generated $R^2$ using permutation of $y$ values. The model considered here is a 3-layered neural net with the class of functions  $\mathcal{F}_{\textrm{NN}}(30,30,30)$. The sample size is 100. The red line denotes the observed $R^2$ for the true pairings of $X$ and $Y$, the green line denotes the 95\%-quantile of the empirical distribution of $R^2$ (approximation using 200 permutations).}
		\label{fig:relation_nn_one}
	\end{subfigure}
	\hfill
	\begin{subfigure}[b]{\textwidth}
		\centering
		\includegraphics[width=\textwidth]{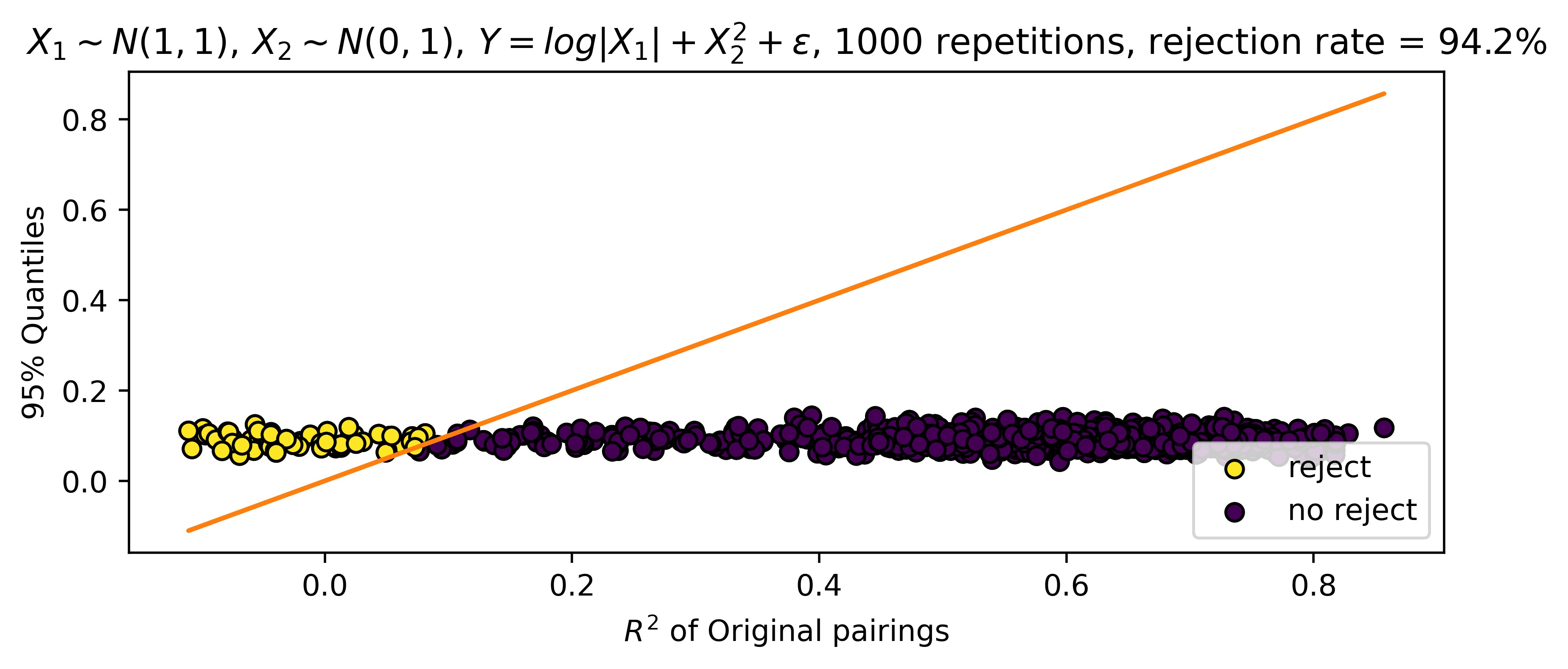}
		\caption{Scatterplot of the $R^2$ values for the original pairings against the 95\% quantiles of the empirical distribution of $R^2$. The orange line shows the identity function.}
		\label{fig:relation_nn_repeats}
	\end{subfigure}
	\caption{Results of the permutation test for $\mathcal{F}_{\textrm{NN}(30, 30, 30)}$ with data generated in a following manner $X_1\sim\mathcal{N}(1,1),X_2\sim\mathcal{N}(0,1)$ and $Y=\log|X_1|+X_2^2+\epsilon$, where $\epsilon\sim\mathcal{N}(0,1)$.}
	\label{fig:relation_nn}
\end{figure}

For the remaining scenarios in this section, we consider only the linear regression model with the class of functions $\mathcal{F}_{\textrm{LR}}$.
We inspect the influence of changing the distribution slightly in the test in order to ensure the statistical analysis using the test is reliable and accurate.
For $a\in\mathbb{R}$ let $X_1\sim\mathcal{N}(a,1),X_2\sim\mathcal{N}(0,0.1)$ be independent and $Y=\log|X_1|+X_2^2+\epsilon$, where $\epsilon\sim\mathcal{N}(0,0.1)$ is the noise.
Consider a sample of size $100$.
Note that the variance of $X_2$ has been decreased in comparison to the previous example.
Only for values of $a$ close to 0, the null hypothesis is not rejected (fig. \ref{fig:normal_mean_small_sd_one}).
This makes sense, since the logarithm changes most rapidly close to $0$ and for those arguments it is difficult to fit a linear function which describes this relationship well.
This pattern is the same with average rejection rate of $H_0$ when repeating the experiment 100 times for each value of $a$, see fig. \ref{fig:normal_mean_small_sd_repeats}.
For values of $a$ greater than 0.6, the $H_0$ is almost never rejected.
When the variance of $X_2$ increases to 0.5, the null hypothesis is no longer rejected for some values of $a$ larger than 5 (fig. \ref{fig:normal_mean_large_sd}).
This particular case shows the influence of available information on rejecting the null hypothesis.
The less informative predictors are the more likely it is not to reject the null hypothesis; we can see that as the parameter $a$ increases, the $\log|X_1|$ becomes flatter slowly losing its predictive value.
Meanwhile, the influence of $X_2^2$ on the value of $Y$ increases and given that the model can only predict linearly in $X_2$, the power of the test decreases.

\begin{figure}
	\centering
	\begin{subfigure}[b]{\textwidth}
		\centering
		\includegraphics[width=\textwidth]{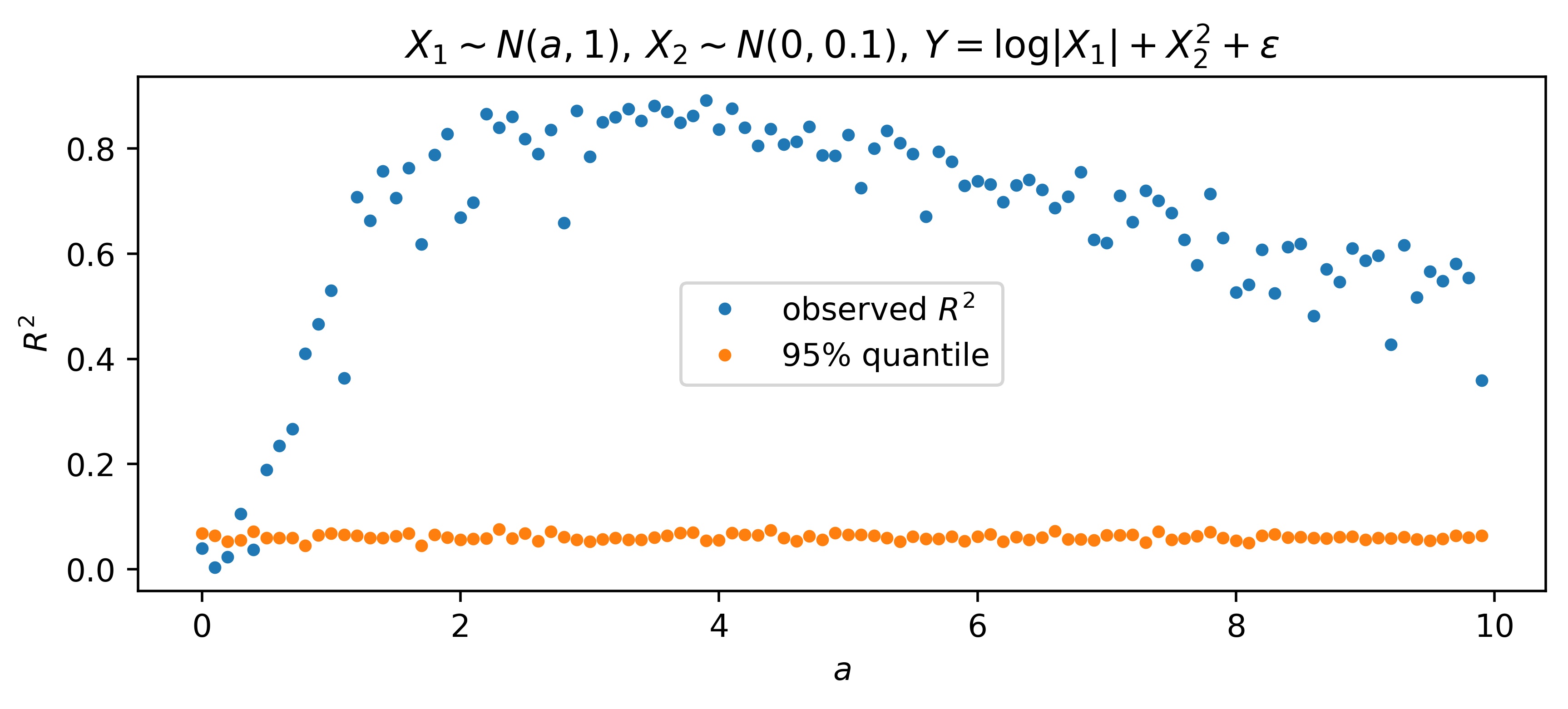}
		\caption{The plot shows the results of performing the permutation test for linear regression model with the class of functions $\mathcal{F}_{\textrm{LR}}$. The sample size is 100. The blue dots show the observed $R^2$ and the orange dots show the 95\% quantile of the empirical distribution of generated $R^2$ (approximation using 200 permutations). The test has been performed for values of $a$ ranging between 0 and 10.}
		\label{fig:normal_mean_small_sd_one}
	\end{subfigure}
	\hfill
	\begin{subfigure}[b]{\textwidth}
		\centering
		\includegraphics[width=\textwidth]{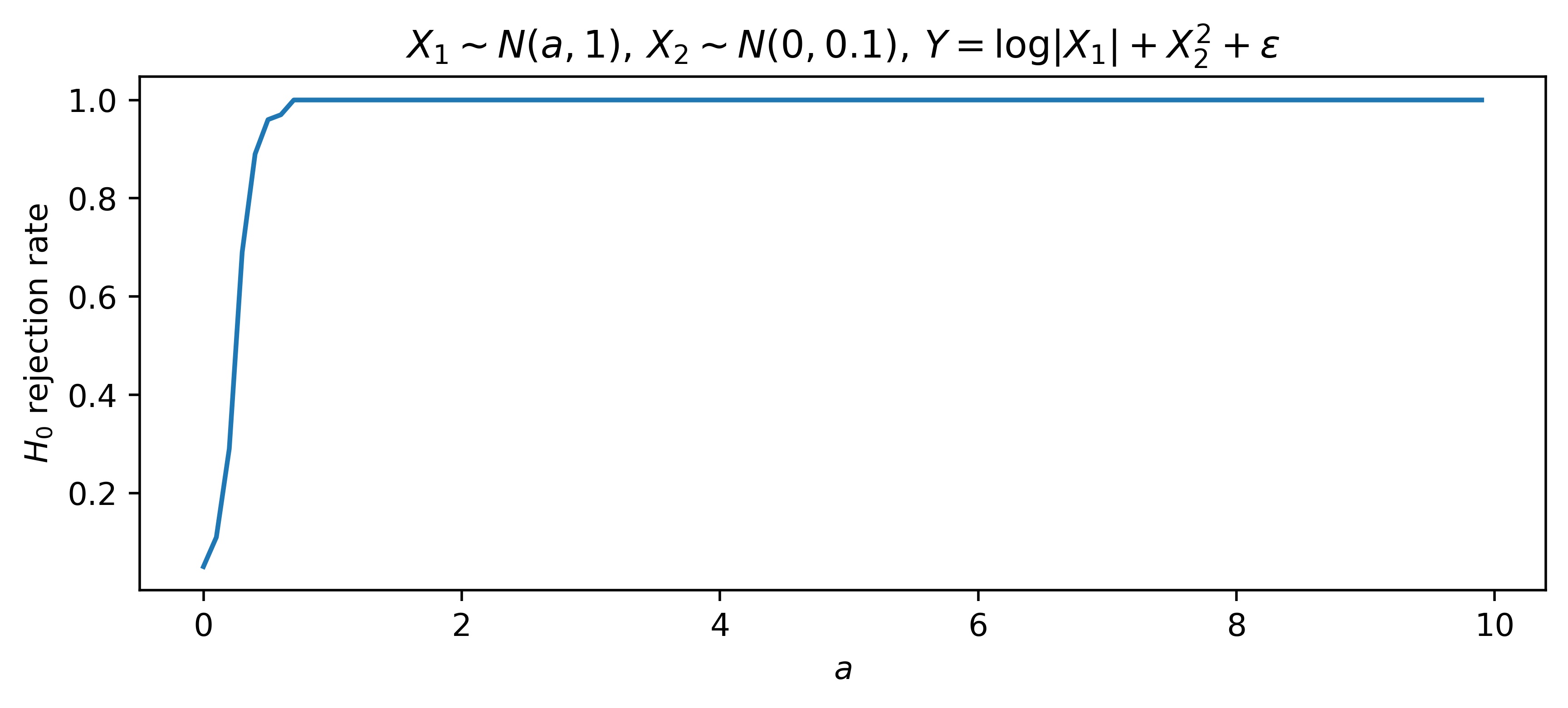}
		\caption{Average rejection rate of $H_0$ with parameter $a$ varying from $0$ to $1$. For each $a$ 100 repetitions were made.}
		\label{fig:normal_mean_small_sd_repeats}
	\end{subfigure}
	\caption{Results of the permutation test for $\mathcal{F}_{\textrm{LR}}$ with data generated in a following manner $X_1\sim\mathcal{N}(a,1),X_2\sim\mathcal{N}(0,0.1)$ and $Y=\log|X_1|+X_2^2+\epsilon$, where $\epsilon\sim\mathcal{N}(0,0.1)$.}
\end{figure}

\begin{figure}
	\centering
	\begin{subfigure}[b]{\textwidth}
		\centering
		\includegraphics[width=\textwidth]{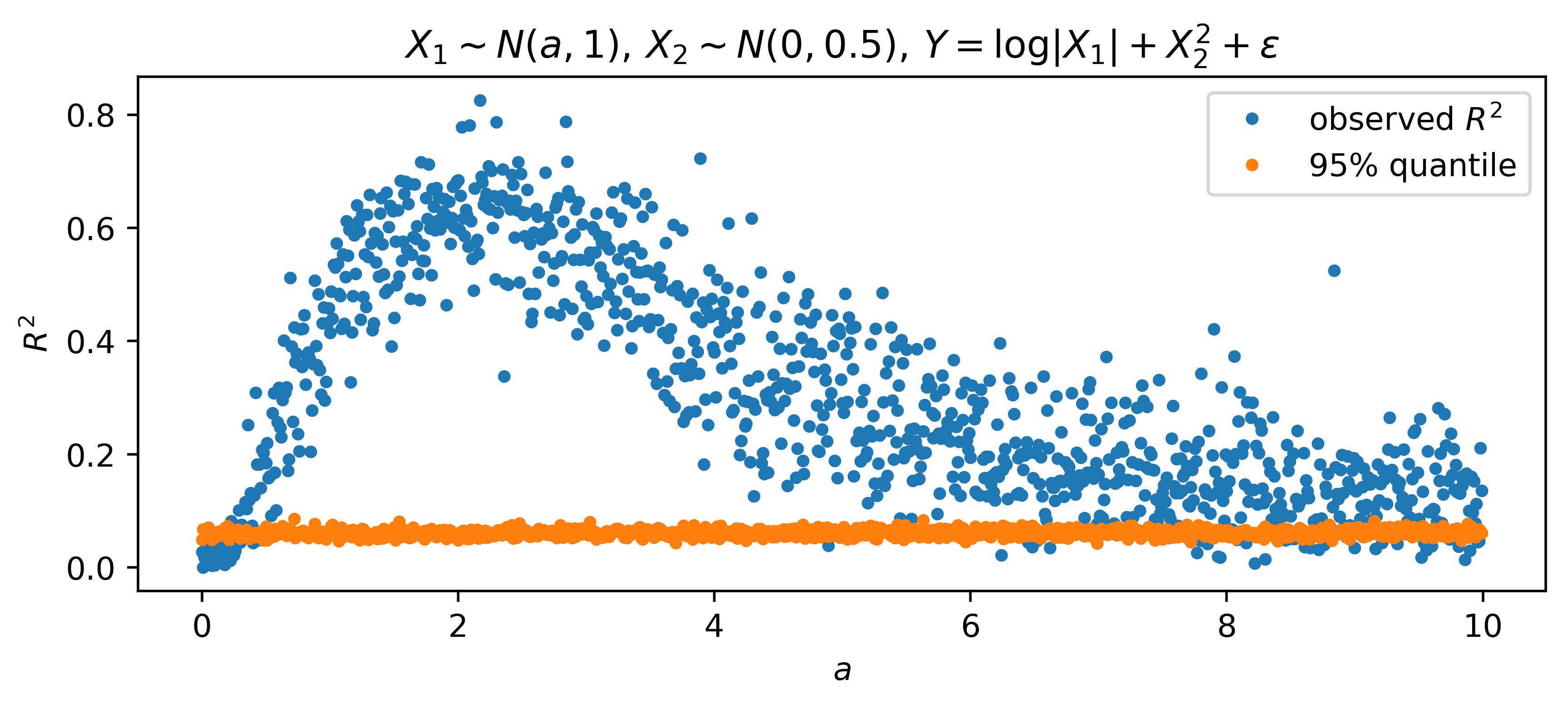}
		\caption{The plot shows the results of performing the permutation test for linear regression model with the class of functions $\mathcal{F}_{\textrm{LR}}$. The sample size is 100. The blue dots show the observed $R^2$ and the orange dots show the 95\% quantile of the empirical distribution of generated $R^2$ (approximation using 200 permutations). The test has been performed for values of $a$ ranging between 0 and 10.}
		\label{fig:normal_mean_large_sd_one}
	\end{subfigure}
	\hfill
	\begin{subfigure}[b]{\textwidth}
		\centering
		\includegraphics[width=\textwidth]{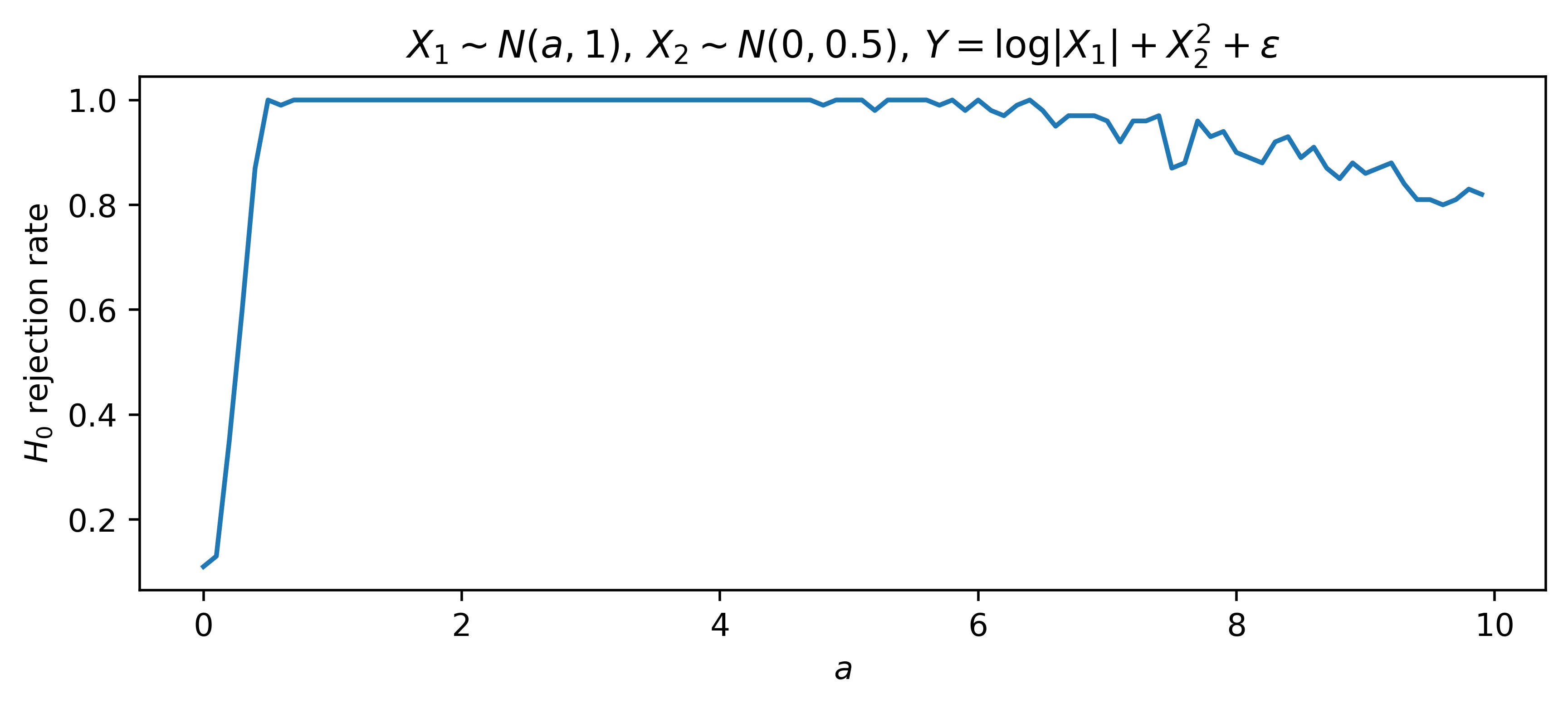}
		\caption{Average rejection rate of $H_0$ with parameter $a$ varying from $0$ to $10$. For each $a$ 100 repetitions were made.}
		\label{fig:normal_mean_large_sd_repeats}
	\end{subfigure}
	\caption{Results of the permutation test for $\mathcal{F}_{\textrm{LR}}$ with data generated in a following manner $X_1\sim\mathcal{N}(a,1),X_2\sim\mathcal{N}(0,0.5)$ and $Y=\log|X_1|+X_2^2+\epsilon$, where $\epsilon\sim\mathcal{N}(0,0.1)$.}
	\label{fig:normal_mean_large_sd}
\end{figure}

Fig. \ref{fig:no_relation_sample_size} and \ref{fig:relation_sample_size} show explicitly the influence of the sample size on the test's capability to reject $H_0$ for the linear regression model with the class of functions $\mathcal{F}_{\textrm{LR}}$.
In the case when $H_0$ is true (fig. \ref{fig:no_relation_sample_size}), the null hypothesis is rejected at a rate of 2-8\% on average regardless of the sample size.\footnote{Again the figure is showing the error rate for Pitman's test as a function of sample size.}
In the case when $H_0$ is false (fig. \ref{fig:relation_sample_size}), specifically with $Y=\log(X)+\epsilon$ for $\epsilon\sim\mathcal{N}(0,1)$, the null hypothesis is rejected much less for smaller sample sizes and the rejection rate increases as the sample size increases reaching close to 95\% at sample size 300.
We can conclude that the power of our test increases until the sample size of around 300, at which point the type II error is particularly low.
Meanwhile, the rejection of a true null hypothesis is rare, even for the smallest of sample sizes.

\begin{figure}
	\centering
	\begin{subfigure}[b]{\textwidth}
		\centering
		\includegraphics[width=\textwidth]{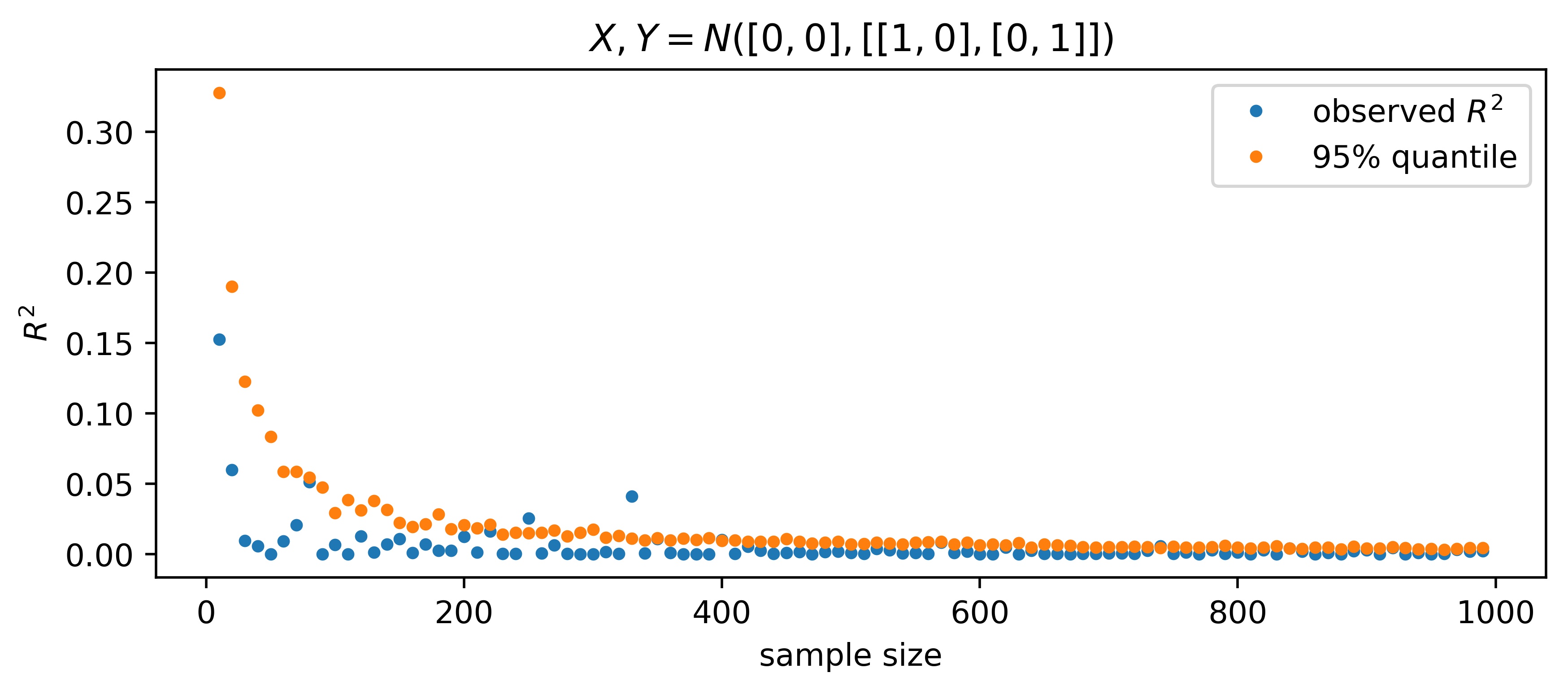}
		\caption{The plot shows the results of performing the permutation test for linear regression model with the class of functions $\mathcal{F}_{\textrm{LR}}$. The blue dots show the observed $R^2$ and the orange dots show the 95\% quantile of the empirical distribution of generated $R^2$ (approximation using 200 permutations). The test has been performed for sample sizes ranging between 10 and 1000.}
		\label{fig:no_relation_sample_size_one}
	\end{subfigure}
	\hfill
	\begin{subfigure}[b]{\textwidth}
		\centering
		\includegraphics[width=\textwidth]{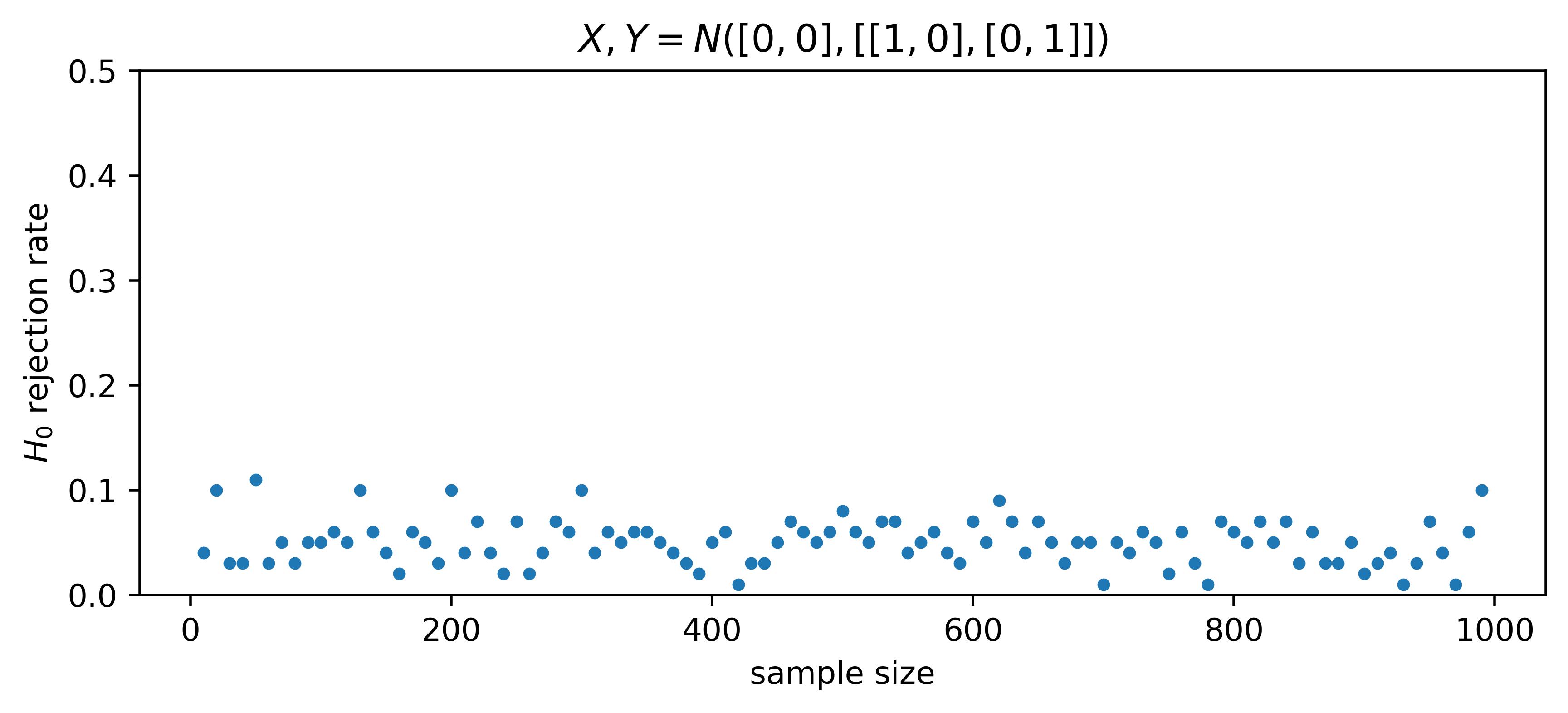}
		\caption{Average rejection rate of $H_0$ with sample size varying from $10$ to $1000$. For each sample size 100 repetitions were made.}
		\label{fig:no_relation_sample_size_repeats}
	\end{subfigure}
	\caption{Results of the permutation test for $\mathcal{F}_{\textrm{LR}}$ with data generated in a following manner $X,Y\sim N(\mu,\Sigma)$, $\mu=\begin{bmatrix}
			0 \\
			0 
		\end{bmatrix}$ and $\sigma=\begin{bmatrix}
			1 & 0\\
			0 & 1
		\end{bmatrix}$.}
	\label{fig:no_relation_sample_size}
\end{figure}

\begin{figure}
	\centering
	\begin{subfigure}[b]{\textwidth}
		\centering
		\includegraphics[width=\textwidth]{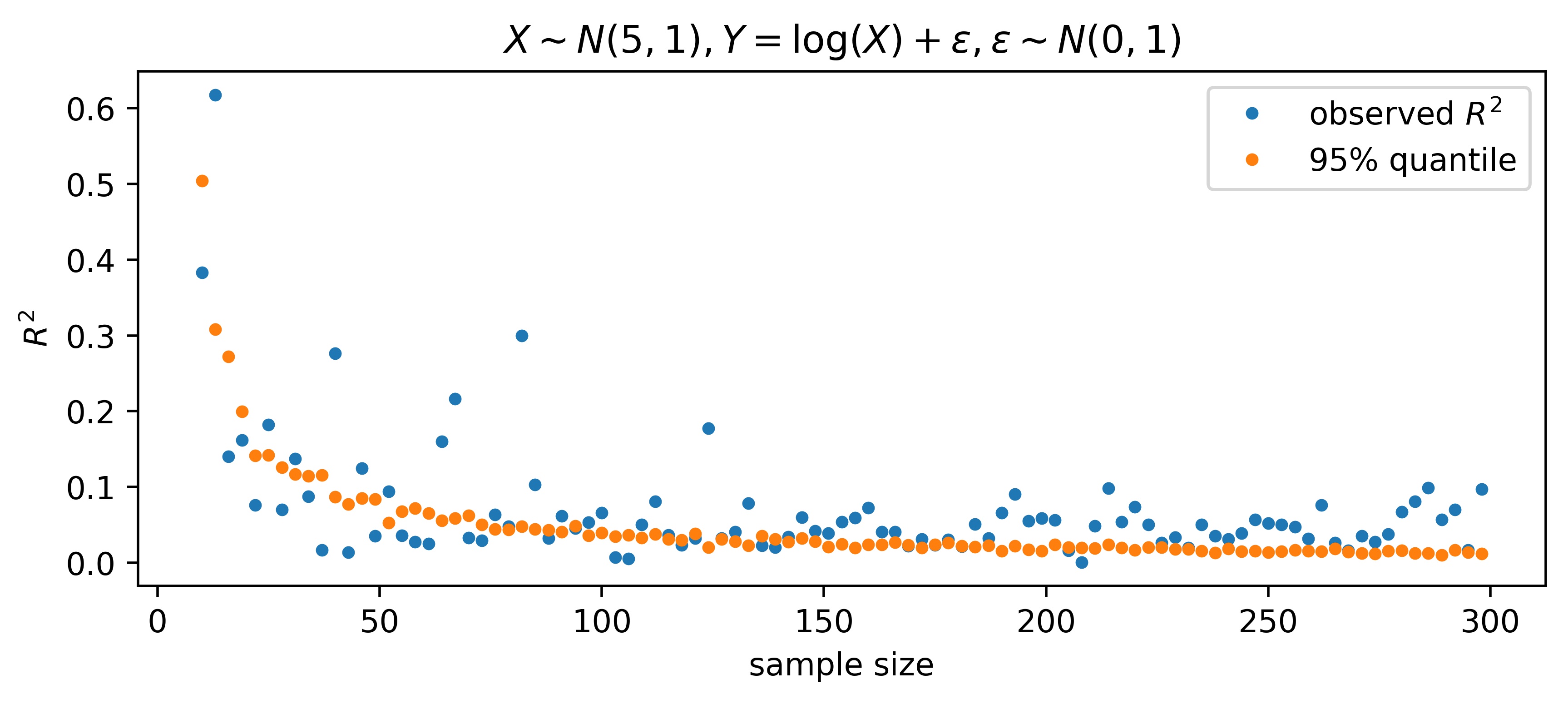}
		\caption{The plot shows the results of performing the permutation test for linear regression model with the class of functions $\mathcal{F}_{\textrm{LR}}$. The blue dots show the observed $R^2$ and the orange dots show the 95\% quantile of the empirical distribution of generated $R^2$ (approximation using 200 permutations). The test has been performed for sample sizes ranging between 10 and 300.}
		\label{fig:relation_sample_size_one}
	\end{subfigure}
	\hfill
	\begin{subfigure}[b]{\textwidth}
		\centering
		\includegraphics[width=\textwidth]{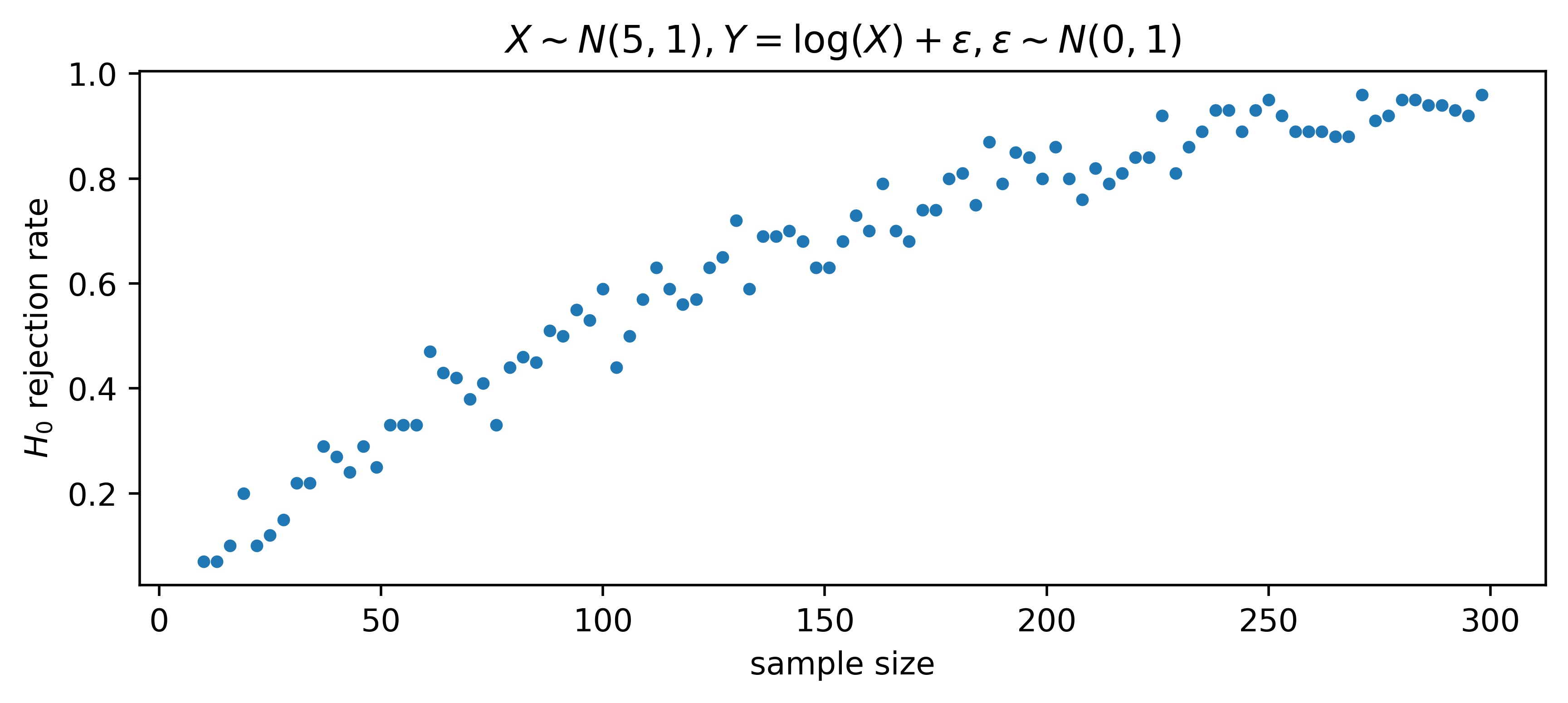}
		\caption{Average rejection rate of $H_0$ with sample size varying from $10$ to $300$. For each sample size 100 repetitions were made.}
		\label{fig:relation_sample_size_repeats}
	\end{subfigure}
	\caption{Results of the permutation test for $\mathcal{F}_{\textrm{LR}}$ with data generated in a following manner $X\sim\mathcal{N}(5,1)$ and $Y=\log|X|+\epsilon$, where $\epsilon\sim\mathcal{N}(0,1)$.}
	\label{fig:relation_sample_size}
\end{figure}

Using the bivariate normal distribution with varying correlation, we can empirically detect the point at which the test rejects $H_0$ for the linear regression model with the class of functions $\mathcal{F}_{\textrm{LR}}$ as the variables become more and more dependent.
Let $0\leq \rho\leq 1$ and $X,Y\sim N(\mu,\Sigma)$, such that
\begin{equation*}
	\mu=\begin{bmatrix}
		0 \\
		0 
	\end{bmatrix},
	\Sigma=
	\begin{bmatrix}
		1 & \rho\\
		\rho & 1
	\end{bmatrix}.
\end{equation*}
Fig. \ref{fig:correlation}  shows that as the correlation reaches 0.3, the test starts to reject $H_0$ almost always in case of sample size $n=100$.\footnote{Note that this figure is showing the error rate for Pitman's test as a function of sample size \cite{pitman37a, pitman37b}.}
We conclude that for sample size $n=100$, the dependence is only detectable reliably by the test when the correlation between variables is greater than 0.3.
This particular example shows that for a given sample size a certain threshold of correlation exists at which the test starts to reject the null hypothesis.
As the correlation increases the rejection becomes more and more likely for a given sample size.

\begin{figure}
	\centering
	\begin{subfigure}[b]{\textwidth}
		\centering
		\includegraphics[width=\textwidth]{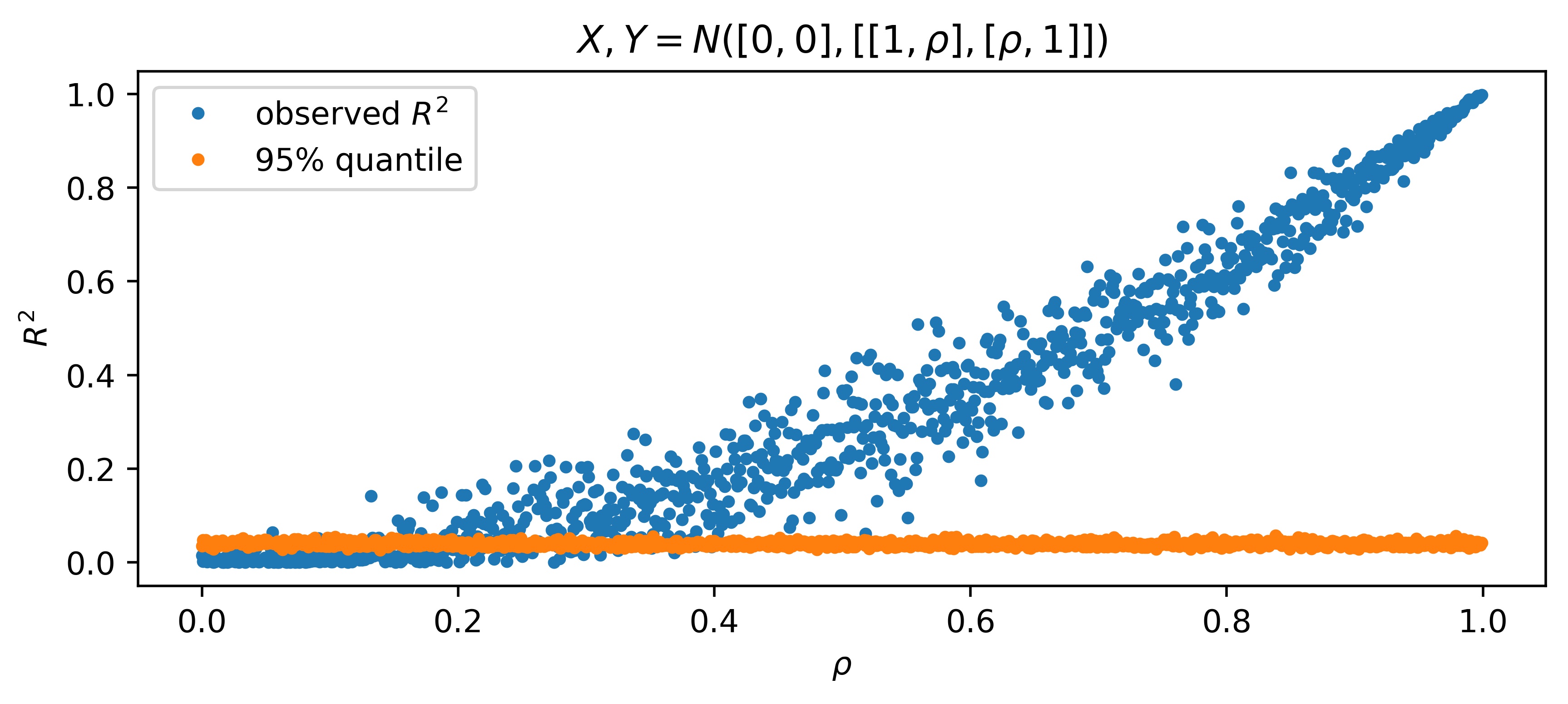}
		\caption{The plot shows the results of performing the permutation test for linear regression model with the class of functions $\mathcal{F}_{\textrm{LR}}$. The sample size is 100. The blue dots show the observed $R^2$ and the orange dots show the 95\% quantile of the empirical distribution of generated $R^2$ (approximation using 200 permutations). The test has been performed for values of correlation $\rho$ ranging between 0 and 1.}
		\label{fig:correlation_one}
	\end{subfigure}
	\hfill
	\begin{subfigure}[b]{\textwidth}
		\centering
		\includegraphics[width=\textwidth]{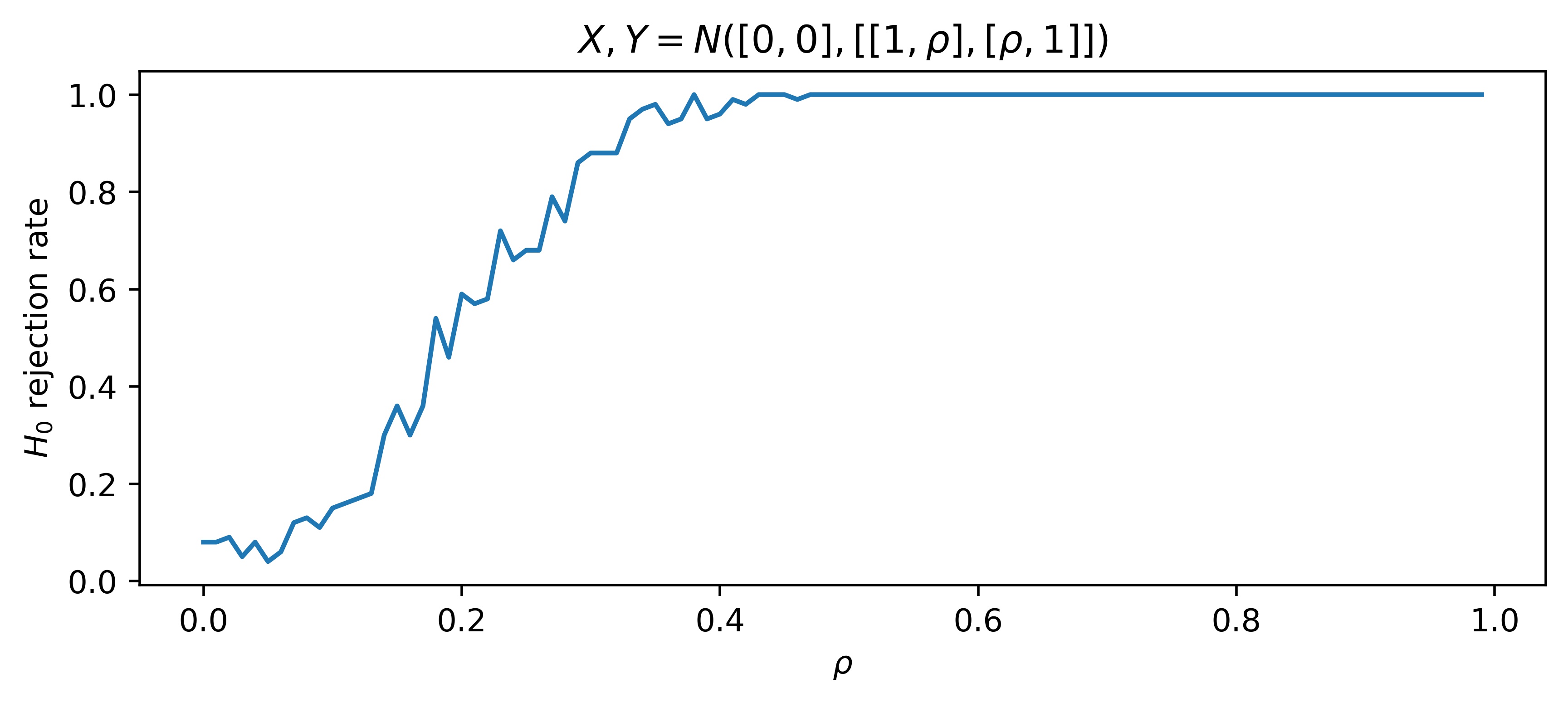}
		\caption{Average rejection rate of $H_0$ with parameter $\rho$ varying from $0$ to $1$. For each $\rho$ 100 repetitions were made.}
		\label{fig:correlation_repeats}
	\end{subfigure}
	\caption{Results of the permutation test for $\mathcal{F}_{\textrm{LR}}$ with data generated in a following manner $X,Y\sim N(\mu,\Sigma)$, $\mu=\begin{bmatrix}
			0 \\
			0 
		\end{bmatrix}$ and $\Sigma=\begin{bmatrix}
			1 & \rho\\
			\rho & 1
		\end{bmatrix}$.}
	\label{fig:correlation}
\end{figure}

Lastly, we present a comparison of our permutation test with a permutation test found in \cite{pesarin10}.
This is also a test for no effect, but specifically in the linear regression model.
Its formulation requires a sample of $n$ i.i.d. observations $\{(X_1,Y_1),\dots,(X_n,Y_n)\}$ from a bivariate variable $(X,Y)$.
We assume that the variables are linked by a linear regression $\mathbb{E}(Y|X=x)=\alpha+\beta\cdot x$, where $\alpha,\beta\in\mathbb{R}$.
The null hypothesis considered for this test is $\beta=0$, under the assumption that responses $Y_i$ can be permuted with respect to covariate $X$.
The test statistic is $T^*_\beta=\sum_i X_iY_i$ and the permutation of $Y_i$ is used when approximating the distribution of the test statistic under $H_0$.
We refer to this test as the permutation test for linear regression after the naming convention in \cite{pesarin10}.
Note that in practice the only difference between our approaches is the choice of test statistic.
In their case, the choice of the test statistic is driven by the null hypothesis.
In our test, the test statistic can be chosen freely as long as it can be calculated using the sample $\{(f(x_i), y_i)\}_i$, which technically means we could use $T^\ast_\beta$ as the test statistic.
In that sense, we can view our test as the generalization of the test for linear regression.

We continue using bivariate normal variables $X$ and $Y$.
We compare the average rejection rate of $H_0$ for both tests with parameter $\rho$ varying from $0$ to $1$.
Fig. \ref{fig:comparison} shows the comparison between the tests.
For sample size $n=100$, the permutation test for linear regression detects the dependence for a slightly smaller $\rho$ than our permutation test, but both reach the rejection rate of 1 at $\rho\approx0.4$.
We conclude that a context-specific test statistic, in this case $T^\ast_\beta$, outperforms more general statistic.
At the same time, our test can use $T^\ast_\beta$ as the test statistic.

\begin{figure}
	\centering
	\includegraphics[width=\textwidth]{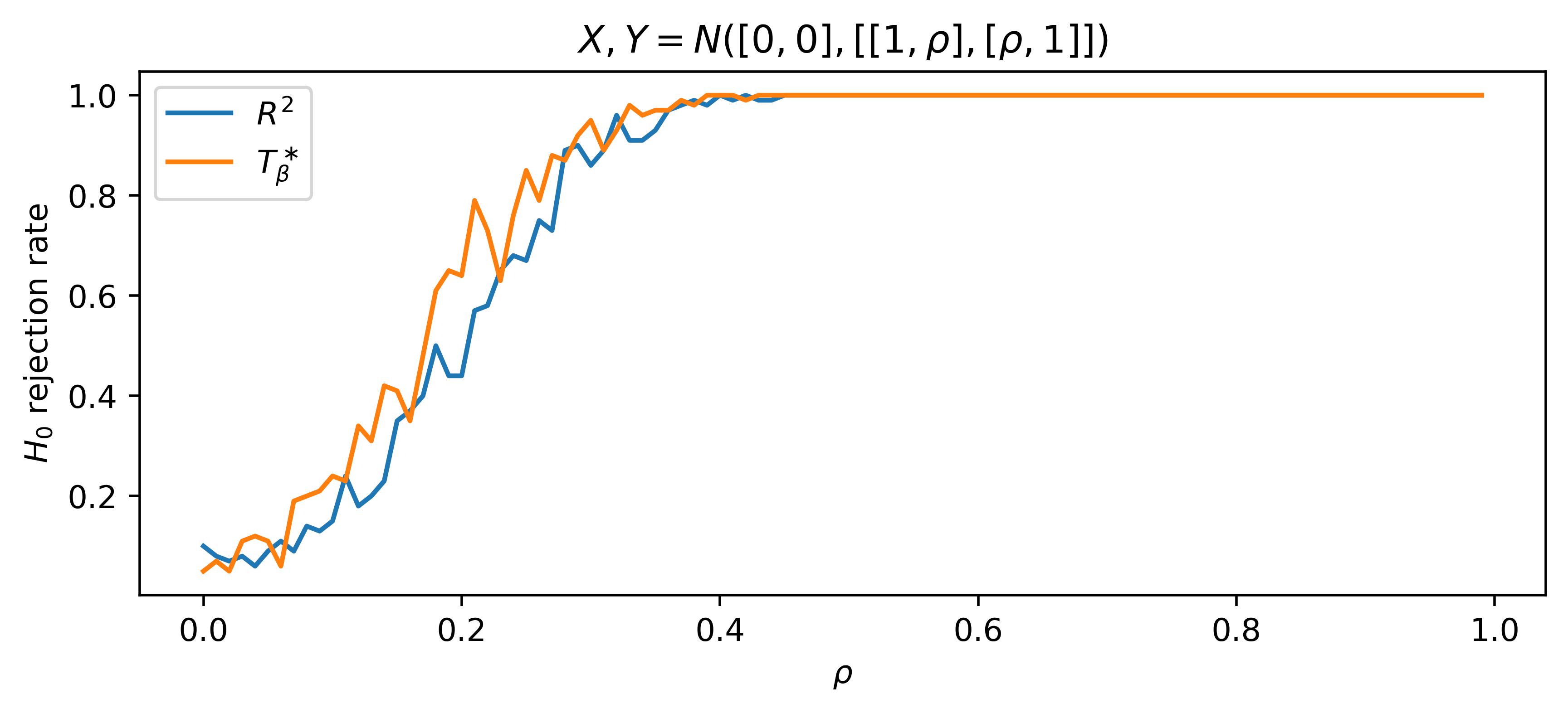}
	\caption{Average rejection rate of $H_0$ with parameter $\rho$ varying from $0$ to $1$. For each $\rho$ 100 repetitions were made. Two different tests were considered. Blue line was generated using $R^2$ as the test statistic, while the orange line was generated using $T^\ast_\beta$ as the test statistic.}
	\label{fig:comparison}
\end{figure}

\end{document}